\documentclass[preprint, 3p, times, twocolumn]{elsarticle}

\usepackage{amssymb}

\usepackage[subtle]{savetrees}

\usepackage{algorithm}
\usepackage[noend]{algpseudocode}

\usepackage{wrapfig}
\usepackage{epstopdf}
\usepackage{framed}

\usepackage{booktabs}
\newcommand{\ra}[1]{\renewcommand{\arraystretch}{#1}}

\usepackage{rotating}
\usepackage{pdflscape}

\usepackage{subcaption}
\usepackage{caption}
\usepackage{fixmath}
\usepackage{amssymb}
\newcommand{\pluseq}{\mathrel{+}\mathrel{\mkern-2mu}=}

\usepackage[table,xcdraw]{xcolor}

\usepackage{hyperref}
\usepackage{xcolor}
\hypersetup{
    colorlinks,
    linkcolor={red!50!black},
    citecolor={blue!50!black},
    urlcolor={blue!80!black}
}

\usepackage{color, colortbl}

\usepackage{amsthm}

\journal{---}

\begin{document}

\begin{frontmatter}




\title{Contraction Clustering (RASTER):\\ A Very Fast Big Data Algorithm for Sequential and Parallel Density-Based Clustering in Linear Time, Constant Memory, and a Single Pass
}




\author{Gregor Ulm\corref{cor1}}
\ead{gregor.ulm@fcc.chalmers.se}
\cortext[cor1]{Corresponding author}

\author{Simon Smith\corref{cor2}}
\ead{simon.smith@fcc.chalmers.se}

\author{Adrian Nilsson\corref{cor2}}
\ead{adrian.nilsson@fcc.chalmers.se}

\author{Emil Gustavsson\corref{cor2}}
\ead{emil.gustavsson@fcc.chalmers.se}
\author{Mats Jirstrand\corref{cor2}}
\ead{mats.jirstrand@fcc.chalmers.se}


\address{Fraunhofer-Chalmers Research Centre for Industrial Mathematics, \\Chalmers Science Park, SE-412 88 Gothenburg, Sweden}
\address{Fraunhofer Center for Machine Learning,\\Chalmers Science Park, SE-412 88 Gothenburg, Sweden}

\begin{abstract}
Clustering is an essential data mining tool for analyzing and grouping similar objects. In big data applications, however, many clustering algorithms are infeasible due to their high memory requirements and/or unfavorable runtime complexity. In contrast, Cont\underline{ra}ction Clu\underline{ster}ing (RASTER) is a single-pass algorithm for identifying density-based clusters with linear time complexity. Due to its favorable runtime and the fact that its memory requirements are constant, this algorithm is highly suitable for big data applications where the amount of data to be processed is huge. It consists of two steps: (1) a contraction step which projects objects onto tiles and (2) an agglomeration step which groups tiles into clusters. This algorithm is extremely fast in both sequential and parallel execution. Our quantitative evaluation shows that a sequential implementation of RASTER performs significantly better than various standard clustering algorithms. Furthermore, the parallel speedup is significant: on a contemporary workstation, an implementation in Rust processes a batch of 500 million points with 1 million clusters in less than 50 seconds on one core. With 8 cores, the algorithm is about four times faster. 

\end{abstract}


\end{frontmatter}


\section{Introduction}
The goal of clustering is to aggregate similar objects into groups in which objects exhibit similar characteristics. However, when attempting to cluster very large amounts of data, i.e.\ data in excess of $10^{12}$ elements~\cite{hathaway2006extending}, two limitations of many well-known clustering algorithms become apparent. First, they operate under the premise that all available data fits into main memory, which does not necessarily hold true in a big data context. Second, their time complexity is unfavorable. For instance, the frequently used clustering algorithms \textsc{DBSCAN}~\cite{ester1996density} runs in $\mathcal{O}(n \log{}n)$ in the best case, where $n$ is the number of elements the input data consists of. There are standard implementations that use a distance matrix, but its memory requirements of $\mathcal{O}(n^2)$ make big data applications infeasible in practice. In addition, the additional logarithmic factor is problematic in a big data context. Linear-time clustering methods exist~\cite{kumar2005linear, kumar2010linear}. Unfortunately, their large factors make them inapplicable to big data clustering. In practice, even an algorithm with a seemingly favorable time complexity can be unsuitable because a large factor cannot just be conveniently ignored and may instead mean that an algorithm is inadequate. Thus, complexity analysis using $\mathcal{O}$-notation can be deceptive as the measured performance difference between two algorithms belonging to the same complexity class can be quite high. Generally speaking, many common clustering algorithms are unable to handle comparatively modest amounts of data and may struggle mightily or even succumb to inputs consisting of just hundreds of clusters and thousands of data points.


In this paper, we introduce Contraction Clustering (RASTER).\footnote{The chosen shorthand may not be immediately obvious: RASTER operates on an implied grid. Resulting clusters can be made to look similar to the dot matrix structure of a \emph{raster graphics} image. Furthermore, the name RASTER is, using our later terminology, an agglomerated contraction of the words \emph{cont\underline{ra}ction} and \emph{clu\underline{ster}ing}. The name of this algorithm thus self-referentially illustrates how it contracts an arbitrary number of larger squares, so-called \emph{tiles}, to single points or, in this case, letters.} In the taxonomy presented in Fahad et al.~\cite{fahad2014survey}, it is a grid-based clustering algorithm. RASTER has been designed for big data. Its runtime scales linearly in the size of its input and has a very modest constant factor. In addition, it is able to handle cases where the available data do not fit into main memory because its memory requirements are constant for any fixed-size grid. The available data can therefore be divided into chunks and processed independently. The two distinctive phases of RASTER are parallelizable; the most computationally intensive first one trivially so, while the second one can be expressed in the divide-and-conquer paradigm of algorithm design. According to Shirkhorshidi et al.~\cite{shirkhorshidi2014big}, distributed clustering algorithms are needed for efficiently clustering big data. However, our algorithm is a counter example to that claim as it exhibits excellent performance metrics even in single-threaded execution on a single machine. It also effectively uses multiple CPU cores. The novelty of RASTER is that it is a straightforward, easy to implement, and extremely fast big data clustering algorithm with intuitive parameters. It requires only one pass through the input data, which it does not need to retain. In addition, key operations like projecting to a tile and neighborhood lookup are performed in constant time.

The remainder of this paper is organized as follows. In Sect.~\ref{sec:problem} we introduce the problem description, with a particular focus on the hub identification problem. This is followed by a presentation of the RASTER algorithm for batch-processing in Sect.~\ref{sec:raster}, which includes the parallel version P-RASTER. An evaluation follows in Sect.~\ref{raster-eval}. As the presentation of RASTER is done partly in a comparative manner, related work is pointed out when appropriate. However, further related work is discussed in Sect.~\ref{sec:related}. This is followed by future work in Sect.~\ref{sec:future} and a conclusion in Sect.~\ref{sec:conclusion}. The appendix contains a qualitative comparison of RASTER versus other common clustering methods in \ref{app-a}, a separate comparison of RASTER and CLIQUE in \ref{app-b}, correctness proofs for P-RASTER in \ref{app-c}, and a detailed analysis of the runtime of P-RASTER in \ref{app-d}.

This paper is a substantial revision of a previously published conference proceedings paper, which focused on a qualitative description of RASTER~\cite{rasterLOD}. In the current paper, we added a substantial quantitative component to evaluate RASTER as well as the new variant RASTER$'$. In addition, we describe and evaluate parallel implementations of these two algorithms.

\section{Problem Description}
\label{sec:problem}
In this section we give an overview of the clustering problem in Sect.~\ref{21}, describe the hub identification problem, which is the motivating use case behind RASTER, in Sect.~\ref{22}, and highlight limitations of common clustering methods in Sect.~\ref{23}.

\subsection{The Clustering Problem}
\label{21}
Clustering is a standard approach in machine learning for grouping similar items, with the goal of dividing a data set into subsets that share certain features. It is an example of unsupervised learning, which implies that there are potentially many valid ways of clustering data points. Elements belonging to a cluster are normally more similar to other elements in it than to elements in any other cluster. If an element does not belong to a cluster, it is classified as noise. An element normally belongs to only one cluster. However, fuzzy clustering methods~\cite{yang1993survey, pham2001spatial} can identify non-disjoint clusters, i.e.~elements may be part of multiple overlapping clusters. Yet, this is not an area we are concerned with in this paper. Instead, our focus is on problems that are in principle solvable by common clustering methods such as \textsc{DBSCAN}~\cite{ester1996density} or $k$-means clustering~\cite{macqueen1967some}.

\subsection{The Hub Identification Problem}
\label{22}


\begin{figure}[h]
\centering
\includegraphics[scale=0.50]{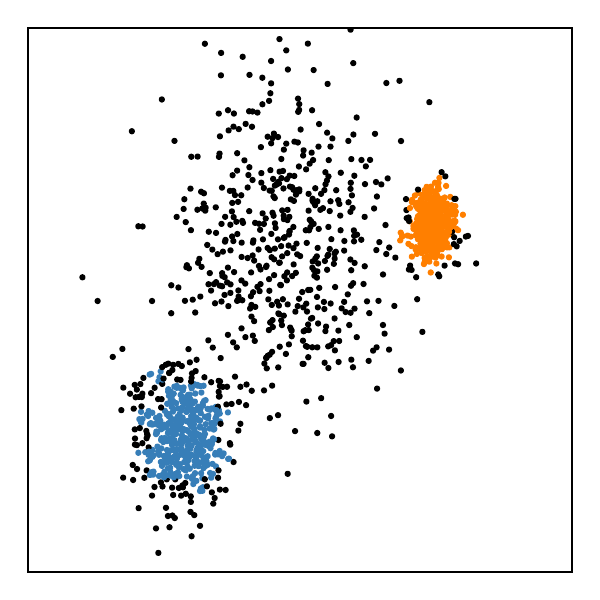}
\caption{This sample illustrates the hub identification problem, where the goal is to find dense clusters in a noisy data set. RASTER identifies two dense clusters and ignores the less dense points in the center.}
\label{fig:RASTER_artificial}
\end{figure}

The main use case RASTER addresses is solving the hub identification problem. Consider being given a vast data set of unlabeled telemetry data in the form of GPS coordinates that trace the movement of vehicles. The goal is to identify \emph{hubs} in vehicle transportation networks. More concretely, hubs are locations within a road network at which vehicles stop so that a particular action can be performed, e.g.~loading and unloading goods at warehouses, dropping of goods at delivery points, or passengers boarding and alighting at bus stops. Conceptually, a hub is a node in a vehicle transportation network, and a node is the center of a cluster. By knowing the location of hubs, as well as the IDs of vehicles that made use of any given hub, detailed vehicle transportation networks can be constructed from GPS data, and usage profiles of vehicles created. A particular feature of a hub is that it consists of a relatively large number of points in a relatively small area. We therefore also refer to them as tight, dense clusters. Generally, they cover an area that ranges from a few square meters to a few dozen square meters; the largest hubs cover a few hundred square meters. Hubs tend to also be a certain minimum distance apart from each other, which implies that if multiple hubs are identified that are close to each other, it is generally the result of a misclassification as they would form only one hub in reality. Figure~\ref{fig:RASTER_artificial}, which shows dense clusters in a noisy data set, illustrates what hubs may look like. 

The provided input to the hub identification problem are GPS coordinates, indicating their latitude and longitude. Their precision is defined by the number of place values after the decimal point. For instance, the coordinate pair $(40.748441, $ $-73.985664)$ specifies the location of the Empire State Building in New York City with a precision of $11.1$ centimeters. A seventh decimal place value would specify a given location with a precision of $1.1$ centimeters, while truncating to five decimal place values would lower precision to $1.1$ meters. High-precision GPS measurements require special equipment, but consumer-grade GPS is only accurate to within about ten meters under open sky~\cite{diggelenGPS}. Altogether, there are three reasons why a reduced precision will lead to practically useful clusters. In addition to the aforementioned imprecision of GPS equipment, one has to take into account vehicle size as well as the size of hubs. Commercial vehicles from where our data originate may have a length ranging from a few meters, in the case of cars or vans, to up to 25 meters or more, in the case of a truck with an attached trailer. Furthermore, hubs cover an area that can reasonably be expected to be much larger than the size of a vehicle. Thus, for the purpose of clustering, lower-precision GPS coordinates could be used, without losing a significant amount of information. In practice, a precision of one meter or less is often sufficient. The example just given may seem to imply that scaling, or reducing the precision, is based on powers of ten. However, this interpretation would not be correct as one could use an arbitrary real-number as a scaling factor as well.

The hub identification problem emerged from working with huge real-world data sets. A common approach of clustering algorithms is to return labels of all input data points, indicating to which cluster they belong. This is also the approach taken by one of our comparative benchmarks (cf.~Sect.\,\ref{raster-eval}). With another linear pass, clusters and their labels can be extracted. This is a viable approach for smaller data sets. However, one of our goals was also to make the amount of data easier to handle and, for instance, reduce $10^{12}$ points to a minuscule fraction of it that contains only the relevant information. Points that do not belong to a cluster are not of interest to our problem. Instead, RASTER retains counts of scaled inputs that are used to determine the approximate area of a cluster in the underlying data. Via reducing the precision of the input, the GPS canvas is therefore divided into a number of \emph{tiles} (cf.~Sect.~\ref{raster-tiles}) for each of which a count of the points it contains is maintained. The number of tiles is finite because the dimensions of the GPS canvas are finite with a latitude $\in [-90, +90]$ and a longitude $\in [-180,180]$. Because the input points are not retained and the range of the input coordinates is finite, the memory requirements of RASTER are constant.



\begin{figure*}
\centering
\begin{subfigure}{.20\textwidth}
  \centering
  \includegraphics[scale=0.50,  trim=1.2cm 2.695cm 2.1cm 0.7cm, clip]{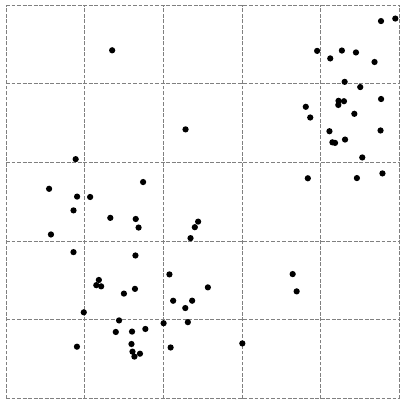}
  \caption{Original data points}
  \label{fig:idea1}
\end{subfigure}\hfill
\begin{subfigure}{.20\textwidth}
  \centering
  \includegraphics[scale=0.50, trim=1.2cm 2.695cm 2.1cm 0.7cm, clip]{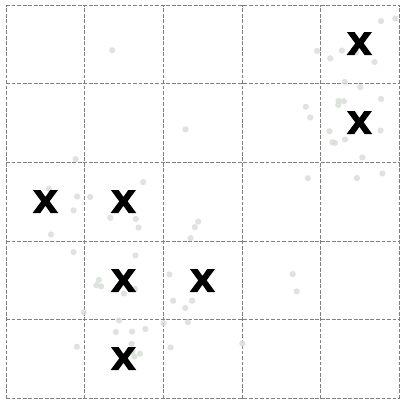}
  \caption{Significant tiles}
  \label{fig:idea2}
\end{subfigure}\hfill
\begin{subfigure}{.20\textwidth}
  \centering
  \includegraphics[scale=0.50, trim=1.2cm 2.695cm 2.1cm 0.7cm, clip]{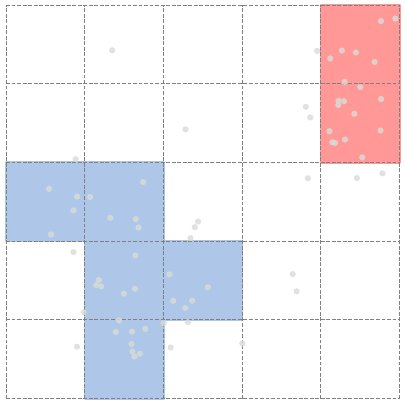}
  \caption{Clustered tiles}
  \label{fig:idea3}
\end{subfigure}\hfill
\begin{subfigure}{.20\textwidth}
  \centering
  \includegraphics[scale=0.50, trim=1.2cm 2.695cm 2.1cm 0.7cm, clip]{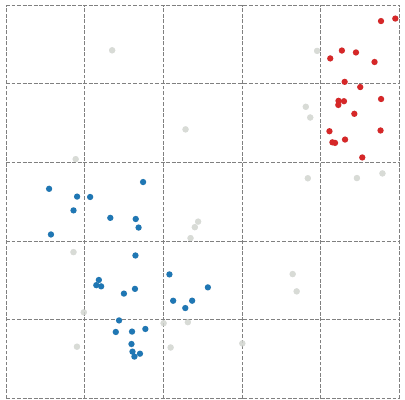}
  \caption{Clustered points}\label{fig:idea4}
\end{subfigure}
\caption{High-level visualization of RASTER (best viewed in color). The original input is shown in a), followed by projection to tiles in b) where only significant tiles are retained. Tile-based clusters are visualized in c), which corresponds to RASTER. Clusters that are returned as collections of points are shown in d), which corresponds to the variant RASTER$'$.}
\label{fig:idea}
\end{figure*}

\subsection{Limitations of two common clustering methods}
\label{23}

In this subsection, we show how the limitations of two common clustering algorithms make them unsuitable for larger data sets, taking into account both runtime and memory considerations. While we focus on showing why the two standard clustering algorithms DBSCAN and $k$-means clustering are not suitable for big data clustering, the bigger point is that these limitations affect all other clustering algorithms that retain their inputs.

DBSCAN identifies density-based clusters. Its time complexity is $\mathcal{O}(n\log{}n)$ in the best case. This depends on whether a query identifying the neighbors of a particular data point can be performed in $\mathcal{O}(\log{}n)$. In a pathological case, or in the absence of a fast lookup query, its time complexity is $\mathcal{O}(n^2)$. \textsc{DBSCAN} is comparatively fast. Yet, when working with many billions or even trillions of data points, clustering becomes infeasible due to the logarithmic factor, provided all data even fits into main memory. We have found that on a contemporary workstation with $16$ GB RAM, the \texttt{scikit-learn} implementation of DBSCAN cannot even handle one million data points. A less efficient implementation using a distance matrix would become infeasible with a much smaller number of data points already. Overall, depending on the implementation, DBSCAN is either too slow or consumes too much memory to be a suitable choice for clustering of big data.



In $k$-means clustering, the number of clusters $k$ has to be known in advance. The goal is to determine $k$ partitions of the input data. Two aspects make $k$-means clustering less suitable for our use case. First, when dealing with big data, estimating a reasonable $k$ is non-trivial. Second, its time complexity is unfavorable. An exact solution requires $\mathcal{O} (n^{dk + 1})$, where $d$ is the number of dimensions~\cite{Inaba1994}. In our case, using $2$-dimensional data, it is $\mathcal{O} (n^{2k + 1})$. Lloyd's algorithm~\cite{lloyd1982least}, which uses heuristics, is likewise not applicable as its time complexity is $\mathcal{O}(dnki)$, where $i$ is the number of iterations until convergence is reached. In practice, $i$ tends to be small and consequently does not pose much of a problem. On the other hand, with increasing numbers of clusters, $k$ is increasingly negatively affecting runtime. There have been recent improvements to improve $k$-means clustering~\cite{bachem2016fast, capo2017efficient}, but their time complexity is still highly unfavorable for huge data sets.




\section{\textsc{RASTER}}
\label{sec:raster}

This section contains a thorough presentation of RASTER, starting with a high-level description of the algorithm in Sect.~\ref{raster-highlevel}, a discussion of the concept of \emph{tiles} as well as their role in creating clusters in Sect.~\ref{raster-tiles}, and an informal time-complexity analysis in Sect.~\ref{raster-alg}. This is followed by an explanation of the variation RASTER$'$, which retains its input, in Sect.~\ref{raster-prime} as well as the parallel variant P-RASTER in Sect.~\ref{sec:p-raster}. Lastly, we highlight some further technical details in Sect.~\ref{raster-further}.


\algrenewcommand\algorithmicrequire{\textbf{input:}}
\algrenewcommand\algorithmicensure{\textbf{output:}}
\begin{algorithm}
\begin{algorithmic}[1]
\Require data \emph{points}, precision \emph{prec}, threshold $\tau$, distance $\delta$, minimum cluster size $\mu$
\Ensure set of clusters \emph{clusters}
\State $\mathit{acc} := \varnothing $ \Comment $\{(x_\pi, y_\pi): \mathit{count}\}$ 
\State $\mathit{clusters} := \varnothing$ \Comment set of sets

\For{$(x, y)$ in \emph{points}} \Comment $\mathcal{O}(n)$ for lls.~3--8   
\State $(x_\pi, y_\pi) := \mathit{project}(x, y, \mathit{prec})$ \Comment $\mathcal{O}(1)$
\If{$(x_\pi, y_\pi)$ $\not\in$ keys of \emph{acc}}{}
\State $\mathit{acc}[(x_\pi, y_\pi)] := 1$
\Else
\State $\mathit{acc}[(x_\pi, y_\pi)] \pluseq 1$
\EndIf
\EndFor

\For{$(x_\pi, y_\pi)$ in $\mathit{acc}$} \Comment $\mathcal{O}(n)$ for lls.~9--11
\If{$\mathit{acc}[(x_\pi, y_\pi)] < \tau$}
\State remove $\mathit{acc}[(x_\pi, y_\pi)]$
\EndIf
\EndFor

\State $\sigma := $ keys of $\mathit{acc}$ \Comment significant tiles

\While{$\sigma \ne\varnothing$} \Comment $\mathcal{O}(n)$ for lls.~12--24
\State $t$ := $\sigma\mathit{.pop()}$
\State \emph{cluster} := $\varnothing$ \Comment set
\State \emph{visit} := $\{t\}$

\While{\emph{visit} $\ne\varnothing$}
  \State $u := $ $\mathit{visit.pop()}$
  \State $\mathit{ns} := \mathit{neighbors}(u, \delta)$  \Comment $\mathcal{O}(1)$
  \State $\mathit{cluster} := \mathit{cluster} \cup \{u\}$
  \State $\sigma := \sigma \setminus \mathit{ns}$  \Comment cf.~ln.~13
  \State $\mathit{visit} := \mathit{visit} \cup \mathit{ns}$
\EndWhile

\If{size of $\mathit{cluster} \geq \mu$}
\State add $\mathit{cluster}$ to $\mathit{clusters}$
\EndIf
\EndWhile

\caption{\textsc{RASTER}}
\label{alg:raster}
\end{algorithmic}
\end{algorithm}

\subsection{High-Level Description}
\label{raster-highlevel}
The goal of \textsc{RASTER} is to reduce a very large number $n$ of $2$-dimensional points to a much more manageable number of coordinates that specify the \emph{approximate area} of clusters in the input data. The input is not retained. Figure~\ref{fig:idea} provides a visualization. In more detail, the algorithm works as follows. It uses an implicit $2$-dimensional grid of a coarser resolution than the input data. Each square of this grid is referred to as a \emph{tile}; each point in the input data is projected to exactly one tile. A tile containing at least a user-specified threshold value $\tau$ of data points is labeled as a \emph{significant tile}. Afterwards, clusters are constructed from adjacent significant tiles. RASTER creates clusters based on tiles; these clusters cover the area in which the clusters of the original input are located (cf.~Fig.~\ref{fig:idea3}). The advantage of this approach is that it provides approximate truth very quickly. Furthermore, for the hub identification problem, the approximate location of a cluster is more important than retaining the precise GPS coordinates of all associated input points; even ignoring points near a cluster that a different clustering method may include is not relevant as long as the cluster and its approximate area are identified. Conversely, a potential downside of RASTER is that it constructs clusters from significant tiles and ignores all points that were projected to those tiles. This leads to great memory efficiency, but it may not provide the user with enough information. After all, the resulting clusters of tiles can be imprecise as they are delimited by the tiles in the implied grid. The variant RASTER$'$ addresses this problem (cf.~Fig.~\ref{fig:idea4}). In addition to clusters and their constituent significant tiles, this variant also retains the original data points that were projected to those tiles. The downside of this increased amount of information is that RASTER$'$ does not use constant memory, unlike RASTER.

   
 

\subsection{Tiles and Clusters}
\label{raster-tiles}
A key component of \textsc{RASTER} is the deliberate reduction of the precision of its input data. This is based on a projection of points to tiles and could, for instance, be achieved by truncating or rounding. The goal is the identification of clusters, which is attained via two distinct and consecutive steps: \emph{contraction} and \emph{agglomeration}. The contraction step first determines the number of input data points per tile and then discards all non-significant tiles. The agglomeration step constructs clusters out of adjacent significant tiles. To illustrate the idea of a \emph{tile}, consider a grid consisting of squares of a fixed side length. A square may contain several coordinate points. Reducing the precision by one decimal digit means removing the last digit of a fixed-precision coordinate. For instance, with a chosen decimal precision of $2$, the coordinates $(1.005, 1.000)$, $(1.009, 1.002)$, and $(1.008, 1.006)$ are all truncated (contracted) to the tile identified by the corner point $(1.00, 1.00)$. Thus, $(1.00, 1.00)$ is a tile with three associated points. This tile would be classified as a significant tile with a threshold value of $\tau \ge 3$ and discarded otherwise. The previous example and the one mentioned in the motivating use case suggest truncating input values, which is implemented as reduction of the precision of the data by an integer power of 10. The input, which consists of floating-point values, is scaled up and converted into integers. However, it is also possible to project input data points to tiles using arbitrary real numbers, except zero. This makes it possible to fine-tune the clustering results or improve performance. For instance, a slightly larger grid size may lead to coarser clusters, but it also entails improved runtime, considering that the number of tiles depends on the chosen precision. With a lower precision value, the number of tiles is much lower, which leads to lowered constant memory requirements and thus an improved runtime of the agglomeration and contraction steps. In general, the number of tiles on a finite canvas of side lengths $a$ and $b$ with precision $p$ is $a \, b \cdot10^{p}$. For practical purposes, $p$ should be the lowest possible value that leads to the desired clustering results.




In the subsequent \emph{agglomeration} step, \textsc{RASTER} clusters are constructed, which consist of significant tiles that are at most a user-specified Manhattan or Chebyshev distance of $\delta$ steps apart. In our implementations, we use a Chebyshev distance of 1, i.e.~we take all eight neighbors of the current tile into account, meaning that significant tiles need to be directly adjacent. Yet, the distance parameter is arbitrary. One alternative would be to use a Manhattan distance $>1$, which implies that a cluster could contain significant tiles without direct neighbors. The parameter $\mu$ specifies the minimum number of tiles a cluster has to contain to be recognized as such by the algorithm. This is of practical importance as it filters out isolated significant tiles that do not constitute a hub. The reasoning behind this approach is that hubs generally cover multiple tiles, so very small clusters are interpreted as noise.

\subsection{Time-Complexity Analysis}
\label{raster-alg}
In this subsection we present some explanations that accompany the RASTER pseudocode in Alg.~\ref{alg:raster}, couched in an informal time-complexity analysis.\footnote{Reference implementations of RASTER and its variants in several programming languages are available at \url{https://github.com/fraunhoferchalmerscentre/raster}.} The algorithm consists of three sequential loops. The first two \emph{for}-loops constitute the \emph{contraction} step and the subsequent \emph{while}-loop the \emph{agglomeration} step. Projecting to a tile consists of associating a data point $p$ to a scaled value representing a tile. RASTER does not exhaustively check every possible tile value, but instead only retains tiles that were encountered while processing data. Due to the efficiency of hash tables, the first \emph{for}-loop runs in $\mathcal{O}(n)$, where $n$ is the number of input data points. Projection is performed in $\mathcal{O}(1)$, which is also the time complexity of the various hash table operations used. After the first \emph{for}-loop all points have been projected to a tile. The second \emph{for}-loop traverses all keys of the hash table \emph{tiles}. Only significant tiles are retained. The intermediate result of this loop is a hash table of significant tiles and their respective number of data points. Deleting an entry is an $\mathcal{O}(1)$ operation. At most, and only in the pathological case where there is exactly one data point per tile, there are $n$ tiles. In any case, it holds that $m \le n$, where $n$ is the number of input data points and $m$ the number of tiles. Thus, this step of the algorithm is performed in $\mathcal{O}(m)$.


\begin{figure*}[t]
\centering
\begin{subfigure}{.33\textwidth}
  \centering
  \includegraphics[scale=0.50, trim=0.5cm 1.3cm 0.5cm 0.5cm, clip]{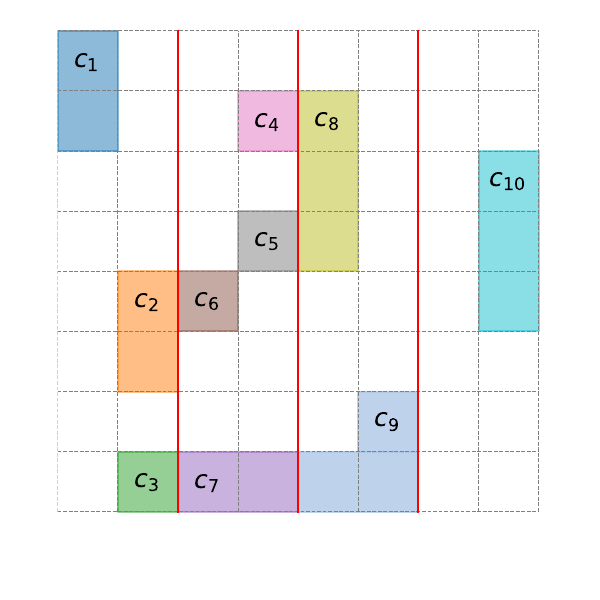}
  \caption{Four subdivisions}
  \label{fig:dc1}
\end{subfigure}\hfill
\begin{subfigure}{.33\textwidth}
  \centering
  \includegraphics[scale=0.50, trim=0.5cm 1.3cm 0.5cm 0.5cm, clip]{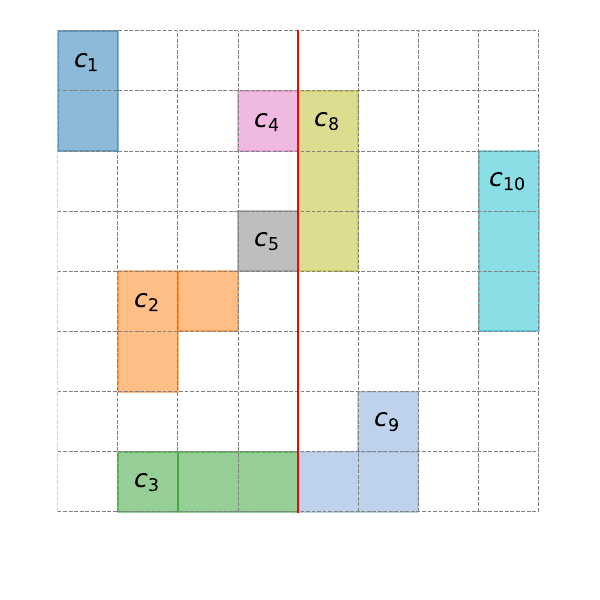}
  \caption{Two subdivisions}
  \label{fig:dc2}
\end{subfigure}\hfill
\begin{subfigure}{.33\textwidth}
  \centering
  \includegraphics[scale=0.50, trim=0.5cm 1.3cm 0.5cm 0.5cm, clip]{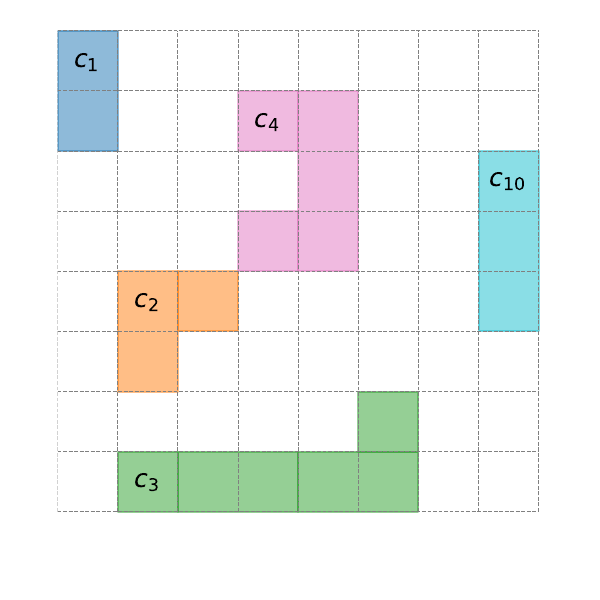}
  \caption{Final clusters}
  \label{fig:dc3}
\end{subfigure}\hfill
\caption{RASTER can be effectively parallelized (figure best viewed in color). This figure shows how clusters that are separated by borders can be joined in parallel in $\log _{}n$ steps, where $n$ is the number of initial slices. In Fig.~\ref{fig:dc1} there are multiple clusters that are separated by borders, which are successively joined as borders between slices are removed. The corner cases of clusters repeatedly crossing a border (e.g.\, $c_4, c_5, c_8$ in Fig.~\ref{fig:dc1}) and clusters crossing multiple borders are shown (cf. clusters $c_5, c_7, c_9$ in Fig.~\ref{fig:dc1}). As joining clusters is not a bottleneck in our use case, the idea presented here remains unimplemented. However, it may be instrumental when implementing RASTER for large-scale data processing (cf.\, Sect.~\ref{sec:future}). Our code iterates through candidate clusters in slices in a sequential manner from right to left; clusters that do not touch the border of a slice are excluded from this step.}
\label{fig:RASTER_DivideAndConquer}
\end{figure*}

The subsequent \emph{while}-loop performs the \emph{agglomeration} step, which constructs clusters from significant tiles. A cluster is a set of the coordinates of significant tiles. There are at most $m$ significant tiles. In order to determine the tiles a cluster consists of, take one tile from the set \emph{tiles} and recursively determine all neighboring tiles in $\mathcal{O}(m)$ in a depth-first manner. Neighborhood lookup with a hash table is in $\mathcal{O}(1)$, as the locations of all its neighboring tiles are known. For instance, when performing agglomerations with a Manhattan distance of 1, the neighbors of any coordinate pair $(x, y)$ are the squares to its left, right, top, and bottom in the grid. One caveat regarding higher dimensional data and how it affects neighborhood lookup should be stated, however. When taking two neighbors per dimension into account, the total number of lookup operations is $2d$, which is suppressed in $\mathcal{O}$-notation. Yet, this implies that RASTER is best-suited for low-dimensional data. In summary, the third loop runs likewise in $\mathcal{O}(m)$. Consequently, the total time complexity is $\mathcal{O}(n)$.


\subsection{The Variation RASTER$'$}
\label{raster-prime}
A key element of RASTER is that information about the input data is only retained in the aggregate and only if it is relevant for clustering. In the projection step, the algorithm determines how many data points were projected to each tile. By discarding the input data, and due to the fact that the range of both the longitude and latitude of GPS data are limited, the resulting memory requirements of our algorithm are constant. Not storing input data is a useful approach when clustering very large amounts of data. However, if all data fits into memory, one could retain all points per tile or all unique points per tile. This increases the memory requirements of the algorithm as clustering can no longer be performed in constant memory. We refer to this variant of our algorithm as RASTER$'$. Compared to the pseudocode in Alg.~\ref{alg:raster}, the required changes are confined to the contraction step. The main difference consists of maintaining a set of unmodified input data points together with each key in the hash table \emph{tiles} instead of incrementing a counter and, subsequently, removing keys if the number of associated stored  data points is below provided threshold value $\tau$. While this variant is not particularly interesting for our motivating use case (cf.~Sect.~\ref{22}), it may allow for a fairer comparison with clustering algorithms that retain their inputs.

\subsection{Parallelizing the Algorithm}
\label{sec:p-raster}


In this section we describe the parallelization of RASTER and RASTER$'$. The resulting parallel versions are called P-RASTER and P-RASTER$'$, respectively. As a preliminary step, however, the algorithm first divides the input into unsorted batches. This is followed by processing each batch in parallel by performing the projection and accumulation steps. For each thread, this results in one hashmap with the number of points per tile. For performance reasons, the input is not sorted, which implies that there can be entries for the same tile in different hashmaps. Consequently, these hashmaps have to be joined. After joining, the significant tiles can be identified and the non-significant tiles removed.

For the next step, the metaphor of a vertical slice of the input space is helpful. The resulting significant tiles are assigned to a slice based on their horizontal position, but they are not sorted within a slice. Afterwards, the significant tiles of each vertical slice are clustered concurrently. However, clusters may cross the border between slices. Thus, only clusters that do not touch any border can be dismissed or retained based on the $\mu$ criterion. The others need to be joined if they have a neighboring tile in a slice to the right or left. They are only discarded if the final joined cluster is below the specified minimum size. This is done while joining slices, which may happen iteratively from one side to the other, or in a bottom-up divide-and-conquer approach. A pseudo-code specification of the agglomeration step of P-RASTER is provided in Alg.~\ref{alg:concurrent-clustering}; the procedure for joining clusters is shown separately as Alg.~\ref{alg:join}. Joining clusters that cross borders has been implemented in a sequential manner as it only consumes a trivial amount of the total runtime.

For slicing the input canvas vertically, we utilize domain knowledge as it is known that the longitude of GPS coordinates is in the range $[-180, 180]$. Thus, we can easily divide the input into evenly spaced slices. Alternatively, one could dynamically determine the minimum and maximum longitude value. Furthermore, it may be more appropriate, also in the context of GPS data, to dynamically determine the width of each slice. With GPS data, this is particularly relevant as the input is normally not uniformly distributed across the input space. For optimal performance, each slice should contain approximately the same number of data points. These performance benefits relate to clustering, but, as stated above, not to the final joining of slices, which happens sequentially. In contrast, Fig.~\ref{fig:RASTER_DivideAndConquer} illustrates how this could be done in parallel with in a divide-and-conquer approach. Lastly, relevant correctness proofs for P-RASTER are provided in \ref{app-c}.

\algrenewcommand\algorithmicrequire{\textbf{input:}}
\algrenewcommand\algorithmicensure{\textbf{output:}}

\begin{algorithm}[t]
\caption{P-RASTER: Agglomeration}
\label{alg:concurrent-clustering}

\begin{algorithmic}[1]
\Require Minimum cluster size $\mu$, parallelism $N$, set of significant tiles
\Ensure Set of clusters $S_c$ with at least $\mu$ tiles
\State Sort all significant tiles into $N$ vertical slices according to their spatial order
\For{each slice in parallel}
    \State initialize $C_b$ \Comment clusters next to border
    \State determine clusters $C$
    \For{$c$ in $C$}
        \If{$c$ touches a border}
            \State add $c$ to $C_b$\Comment{candidates for joining}
        \ElsIf{$|c| \geq \mu$}
            \State add $c$ to $S_c$
        \EndIf
    \EndFor
\EndFor \label{alg:line:done-clustering}
\For{border $B_i$ in $B_{N-1}...B_1$} \Comment{iterative joining, right to left}
    \State $C_{lr} := $ clusters in $C_b$ that touch border $B_i$
    \State $C_j, C_{inter} :=$ \Call{join\_border}{$C_{lr}$} \Comment{cf. Alg.~\ref{alg:join}}
    \State add $C_j$ to $S_c$
    
\State remove clusters that touch $B_i$ and $B_{i-1}$ from $C_b$ 
\State add $C_{inter}$ to $C_b$
\EndFor

\end{algorithmic}
\end{algorithm}

\begin{algorithm}[h]
\caption{P-RASTER: Joining clusters}
\label{alg:join}

\begin{algorithmic}[1]
\Require Minimum cluster size $\mu$, clusters $C_{lr}$ that touch border $B_i$
\Ensure Two sets of clusters, $C_j$ and $C_{inter}$. The first set contains clusters with at least $\mu$ tiles that only touch border $B_{i}$. The second set contains clusters that touch both $B_{i}$ and $B_{i-1}$.
\Procedure{join\_border}{$C_{lr}$}
\State initialize the sets $C_{inter}$ and $C_j$
\For{$c$ not visited in $C_{lr}$} \label{alg:line:start}
    \State $c_{join} := c$
    \State $to\_visit := \{c\}$
    \For{$v$ in $to\_visit$}
        \State take neighbors to $v$ from $C_{lr}$
        \State add neighbors to $to\_visit$
        \State merge neighbors into $c_{join}$
    \EndFor
    \If{$c_{join}$ touches $B_{i-1}$}
        \State add $c_{join}$ to $C_{inter}$ \Comment{add inter-slice cluster}
    \ElsIf{$|c_{join}| \geq \mu$}
        \State add $c_{join}$ to $C_j$
    \EndIf
\EndFor
\EndProcedure
\end{algorithmic}
    
\end{algorithm}

\subsection{Further Technical Details}
\label{raster-further}
This subsection contains further technical details, covering the practical aspect of working with arbitrarily large data sets, the mostly theoretical issue of disadvantageous grid layouts, and a brief remark on higher-dimensional input data.


\paragraph{Working with huge data sets}
RASTER can handle arbitrarily large data sets as its memory requirements are constant. To demonstrate this, assume $M$ available main memory of which $T < M$ is required for the hash table \emph{tiles}. The values stored in it are fixed-size integers and the number of tiles has a constant upper bound as the range of the input values is fixed. This implies that $T$ has a fixed upper bound. Afterwards, partition the input data into chunks of at most size $P \leq M - T$ and iteratively perform the projection step on them. After processing the entire input, continue with determining significant tiles, followed by clustering. With this approach, RASTER can not only process an arbitrarily large amount of data, the input can also be partitioned in an arbitrary manner and processed in an arbitrary order without affecting the resulting clusters. 



\paragraph{Minimum Cluster Size in Disadvantageous Grid Layouts}
We consider truncation of a fixed number of decimal digits to be the standard behavior of \textsc{RASTER}. As long as the projection to tiles is injective, any projection can be chosen. Yet, for any possible projection a corner case can be found that illustrates that a significant tile may not be detected based on the particular implied grid of the chosen projection. However, due to hubs generally consisting of a large number of points, this is merely a theoretical issue. As an illustration, assume a threshold of $\tau = 4$ for significant tiles, and four adjacent points. If all points were located in the same tile of a grid, a significant tile would be detected. However, those four points could also be spread over neighboring tiles, as illustrated by Fig.~\ref{fig:RASTER_Limit}. One could shift the grid by choosing a different projection, but an adversary could easily place all points on different tiles in the new grid. In order to alleviate this problem, a threshold $\tau' < \tau$ for the number of data points in a tile would need to be picked. In the provided example, for instance, a value of $\tau' = \frac{\tau}{4}$ would be required, but his would mean that every tile is significant as long as at least one point was projected to it. An alternative and arguably preferable solution to this problem would be to add an additional step to RASTER to make the results more precise. With a complexity of $\mathcal{O}(m)$, where $m$ is the number of tiles, one could retain any group of four tiles containing at least $\tau$ points. This is can be done in linear time because, for each row in a grid, one only needs to take the current and next row into account. For example, one could start with the tile in the top left corner of the input space and take its right neighbor as well as the two tiles adjacent in the next row into account. This is followed by iterating through the entire grid in that manner.

\begin{figure}
\centering
\includegraphics[scale=0.40]{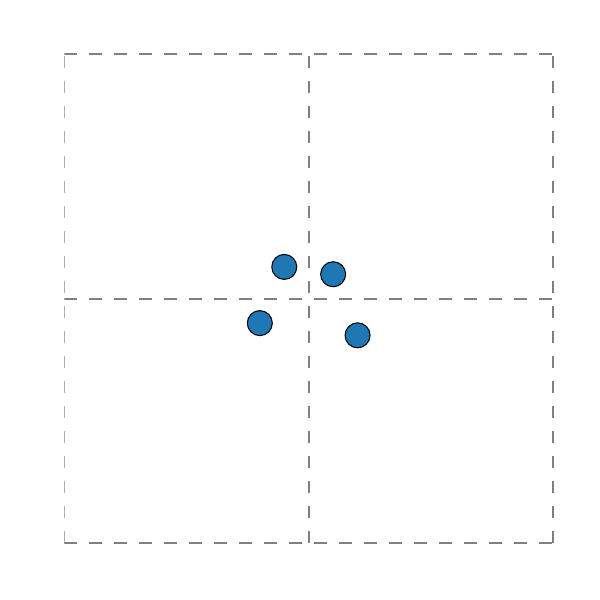}
\caption{RASTER is based on the idea of significant tiles, i.e.\ tiles that contain more than a threshold $\tau$ number of points. However, as this figure illustrates, the location of the grid can interfere and separate points in close vicinity. In the given example, when using a threshold of $\tau = 4$, no significant tile would be detected. This is only a theoretical problem for our algorithm and its primary use case as hubs contain large numbers of points. However, a post-processing step can solve this minor issue at a very modest cost.}
\label{fig:RASTER_Limit}
\end{figure}

\paragraph{Generalizing to Higher Dimensions}
While we have focused on $2$-dimensional geospatial data, it is possible to generalize RASTER to $d$ dimensions. Generalizing from $\mathbb{R}^2$, irrespective of dimension, a similar case can be constructed for $\mathbb{R}^n$. Considering the special case of a projection based on powers of 10 (cf.~Sect.~\ref{raster-tiles}), the reduction for each decimal value is $10^2$ and $10^n$ per tile for $\mathbb{R}^2$ and $\mathbb{R}^n$, respectively. In the case of a Manhattan distance of $\delta = 1$, the number of neighbors to consider per tile is $2d$. When taking all neighbors into account, e.g.~8 in the case of two dimensions, the number of neighbors to look up per tile is $3^d-1$.


\section{Evaluation}
\label{raster-eval}
In this section, we present an evaluation of RASTER based on sequential processing of increasingly larger inputs. Afterwards, we asses how stand-alone versions of RASTER implemented in Python and Rust scale with increasing inputs. This is followed by an evaluation of a parallel implementation of RASTER.

\subsection{Experiments}
\label{eval-experiment}
We subjected RASTER to three experiments. Experiment 1 is based on an existing benchmark that is part of the \texttt{scikit-learn} library, version 0.20.3, comparing several common clustering algorithms with comparatively modestly sized inputs. Experiment 2 explores how implementations in RASTER and RASTER$'$ scale with different parameter values and input sizes. Lastly, experiment 3 explores the parallelization potential of RASTER.

For experiment 1 we were aided by the \texttt{scikit-learn} Python implementations of various clustering algorithms. In this comparative evaluation of RASTER, we use the following nine standard clustering algorithms of this library: 
mini-batch $k$-means clustering~\cite{sculley2010web},
affinity propagation~\cite{frey2007clustering},
mean-shift clustering~\cite{comaniciu2002mean},
spectral clustering~\cite{shi2000normalized, stella2003multiclass},
Ward's method~\cite{ward1963hierarchical},
agglomerative clustering~\cite{voorhees1986implementing},
DBSCAN~\cite{ester1996density},
BIRCH~\cite{zhang1996birch},
Gaussian mixture models~\cite{Reynolds2009GaussianMM}. In addition, we compare RASTER to CLIQUE~\cite{Agrawal1998, Agrawal2005}. All of these algorithms can be used for density-based clustering. However, some were designed for vastly different use cases. We consider these to be the practically most relevant clustering algorithms, which is a view that is supported by the maintainers of \texttt{scikit-learn}. This open-source project uses a relatively high threshold for inclusion that considers, among others, how much interest an algorithm has garnered in research and how widely it has been adopted by practitioners.\footnote{The maintainers state, as part of the documentation of \texttt{scikit-learn}, "We only consider well-established algorithms for inclusion. A rule of thumb is at least 3 years since publication, 200+ citations, and wide use and usefulness. A technique that provides a clear-cut improvement (e.g.~an enhanced data structure or a more efficient approximation technique) on a widely-used method will also be considered for inclusion." See~\url{https://scikit-learn.org/stable/faq.html\#what-are-the-inclusion-criteria-for-new-algorithms} (Accessed 18 December 2024).} CLIQUE is the only algorithm not directly taken from \texttt{scikit-learn}. The reason for adding CLIQUE was that there are some conceptual similarities between this particular algorithm and RASTER. Consequently, we consider these ten algorithms the most relevant for our evaluation. 

The implementation of CLIQUE is a modification of an open-source implementation.\footnote{The code of CLIQUE is based on the implementation by Gy{\"o}rgy Katona, which we sped up by a factor higher than $3$. The original code is available at \url{https://github.com/georgekatona/clique} (Accessed 15 April 2019). Our modified code is included in the code repository accompanying this paper.} RASTER is severely hamstrung in this experiment as the particular setup of this benchmark necessitates us to perform the projection step a second time over the entire input set in order to generate the expected output. In contrast, our algorithm was designed to perform a single pass and not retain its input. We therefore label this variant of our algorithm RASTER$\star$. All algorithms were tweaked to give good results, i.e.\ detecting as many clusters as possible. In this experiment, the input consists of increasingly large input files, ranging from 5K to 500K data points, corresponding to 10 to 1,000 clusters. Each algorithm implementation was executed a total of ten times, recording the number of clusters identified, the runtime in seconds, and the average silhouette coefficient~\cite{ROUSSEEUW198753}, denoted as $s$, which is an indicator of the quality of the resulting clustering based on cluster cohesion (the higher, the better) and cluster separation (the higher, the better). As the input consists of dense and spatially separated clusters, the resulting value should be 1.00 or close to it. This metric thus serves as a plausibility check.

The focus is on comparing sequential performance, but some of the Python implementations use more than one CPU core. As this does not affect RASTER, which runs on a single CPU core, it was not of a primary concern for us. Yet, this means that the performance of some of the other algorithms is arguably better than it would have been if it was possible to enforce using only a single core. In addition to measuring runtime and cluster quality, we also explored at which points the chosen algorithms run out of memory.




\begin{table*}[t]
\centering\ra{1.0}\begin{tabular}{@{}lrrrrrrrrr@{}}

\toprule \emph{Clusters} \phantom{} &


 \textbf{\%} & \textbf{$t$} & \phantom{abC} \textbf{$s$}
&& \phantom{abcdef} \textbf{\%} & \textbf{$t$} & \phantom{abc} \textbf{$s$}\\

 \toprule

& \multicolumn{3}{r}{\textbf{RASTER$\star$}} & & &\\
$10^1$ &  100.0 & 0.02 $\pm$ \hphantom{0}0.00  & 1.00
       &&  & & &&\\
$10^2$ & 100.0 & 0.19 $\pm$ \hphantom{0}0.01 & 1.00
       && & & &&\\
$10^3$ & 100.0 & 1.93 $\pm$ \hphantom{0}0.01 & 1.00
       && & & &&\\

\midrule

& \multicolumn{3}{r}{\textbf{Mean Shift}} & &
  \multicolumn{3}{r}{\textbf{DBSCAN}}\\
$10^1$ &  100.0 & 0.03 $\pm$ \hphantom{0}0.00 & 1.00 
       && 100.0 & 0.03 $\pm$             0.00 & 1.00\\
$10^2$ &  100.0 & 0.36 $\pm$ \hphantom{0}0.01 & 1.00
       && 100.0 & 0.42 $\pm$             0.03 & 1.00\\
$10^3$ &  100.0 & 5.05 $\pm$ \hphantom{0}0.06 & 1.00
       &&  99.8 & 6.12 $\pm$             0.12 & 1.00 \\
\midrule

& \multicolumn{3}{r}{\textbf{Mini-Batch $k$-Means}} & &
  \multicolumn{3}{r}{\textbf{BIRCH}}
\\
$10^1$ & 100.0 & 0.04 $\pm$ \hphantom{0}0.00 & 1.00
       && 100.0 & 0.09 $\pm$ 0.05 & 1.00\\
$10^2$ & 100.0 & 1.22 $\pm$ \hphantom{0}0.63 & 1.00
       && 100.0 & 1.52 $\pm$ 0.00 & 1.00\\
$10^3$ & 100.0 & 4.97 $\pm$ \hphantom{0}0.05 & 0.93
       && 99.0 & 17.08 $\pm$ 0.60 & 0.99\\
\midrule

& \multicolumn{3}{r}{\textbf{Gaussian Mixture}} & &
\multicolumn{3}{r}{\textbf{Ward}}
\\
$10^1$ &  100.0 & 0.03 $\pm$ \hphantom{0}0.00 & 1.00
       && 100.0 & 0.34 $\pm$ 0.00 & 1.00\\
$10^2$ &  100.0 & 1.31 $\pm$ \hphantom{0}0.12 & 1.00
       && 100.0 & 47.39 $\pm$ 9.38 & 1.00\\
$10^3$ &  100.0 & 164.92 $\pm$ \hphantom{0}4.79 & 1.00
       && --- & --- & ---\\

\midrule

& \multicolumn{3}{r}{\textbf{Agglomerative}} & &
\multicolumn{3}{r}{\textbf{CLIQUE}}
\\
$10^1$ &  100.0 & 0.34 $\pm$ \hphantom{0}0.00 & 1.00
       && 100.0 & 0.28 $\pm$ 0.00 & 1.00\\
$10^2$ &  100.0 & 52.34 $\pm$ 10.56 & 1.00
       && 100.0 & 312.26 $\pm$ 9.15 & 0.97\\
$10^3$ & --- & --- & ---
       && --- & --- & ---\\

\midrule

& \multicolumn{3}{r}{\textbf{Spectral}} & &
\multicolumn{3}{r}{\textbf{Affinity Propagation}}
\\
$10^1$ & 100.0 & 1.87 $\pm$ \hphantom{0}0.06 & 1.00
       && 100.0 & 22.07 $\pm$ 0.97 & 1.00\\
$10^2$ & --- & --- & ---
       && --- & --- & ---\\
$10^3$ & --- & --- & ---
       && --- & --- & ---\\
\bottomrule\end{tabular}
\caption{Comparison of sequential RASTER with various standard clustering algorithms. RASTER is labeled RASTER$\star$ because it deviates from our specification. In order to integrate with the \texttt{scikit-learn} benchmark, this algorithm has to retain the input in memory and read it two times. The input size is stated as the number of clusters, where each cluster contains 500 points. The percentage given is based on the number of clusters that are identified versus the number generated as the input. All times $t$ are rounded, so values below 0.004 seconds are recorded as 0.00. As a qualitative performance metric, the silhouette coefficient $s$ was used. Missing entries are due to algorithms running out of main memory or excessive runtime.} 
\label{table:comparison}
\end{table*}

The comparative benchmark in experiment 1 puts our algorithm at a significant disadvantage due to the fact that RASTER$\star$ is slower than RASTER. We therefore performed an additional experiment that gives a better insight into the true performance of our algorithm: Experiment 2 evaluates stand-alone implementations of RASTER and RASTER$'$ in Python and Rust. Unlike with the previous implementations, this allowed us to closely follow the specification provided in Alg.~\ref{alg:raster}. A key aspect of RASTER is that input data points are projected to tiles. The number of implied tiles is based on a precision factor, and the higher that factor is, the greater the number of tiles is as well. With a lower precision value, RASTER is faster and requires less memory as there are fewer tiles to project to, but resulting clusters are coarser. The algorithms processed input files with $10^2$ to $10^6$ clusters, corresponding inputs ranging from 50K to 500M data points. For the chosen precision values of  3, 3.5, 4, and 5, we recorded the number of clusters identified and the runtime. This experiment also serves as a baseline for a comparison with parallel RASTER, which is described below.

The goal of experiment 3 is to empirically determine the speedup of our parallelization efforts. We benchmarked Rust implementations of RASTER and RASTER$'$ on large data sets, one with 50M points ($10^5$ clusters) and one with 500M data points ($10^6$ clusters). For precision values of 3.5 and 4, we measured the total runtime as well as the time spent on projection and clustering as the number of cores increases from 1 to 8 in powers of 2.

All experiments were executed on a contemporary workstation with an octa-core AMD Ryzen 7 2700X CPU, clocked at 3.70 GHz, and 64 GB RAM. The chosen operating system was CentOS Linux 7.6.1810. 

\definecolor{Gray}{gray}{0.9}

\begin{table*}[t]
\centering\ra{1.0}\begin{tabular}{@{}lrrrrrrrrrrr@{}}

\toprule&

\multicolumn{2}{c}{prec = 3} & \phantom{i}&
\multicolumn{2}{c}{prec = 3.5} & \phantom{i}& \multicolumn{2}{c}{prec = 4} & \phantom{i}& \multicolumn{2}{c}{prec = 5}
\\

\cmidrule{2-3} \cmidrule{5-6} \cmidrule{8-9} \cmidrule{11-12}
\emph{Clusters}& \phantom{} \% & $t$
&& \phantom{} \% & $t$ 
&& \phantom{} \% & $t$
&& \phantom{} \% & $t$

\\ \midrule

&  &

\multicolumn{9}{c}{\textbf{RASTER (Python)}}
\\
\rowcolor{Gray}
$10^2$
&   87.0 & 0.03 $\pm$ 0.00
&& 100.0 & 0.03 $\pm$ 0.00
&& 100.0 & 0.04 $\pm$ 0.00
&&  87.0 & 0.04 $\pm$ 0.00\\

$10^3$
&   77.6 & 0.35 $\pm$ 0.00
&& 100.0 & 0.35 $\pm$ 0.01
&& 100.0 & 0.42 $\pm$ 0.00
&&  93.7 & 0.41 $\pm$ 0.00\\

\rowcolor{Gray}
$10^4$
&   78.7 & 3.69 $\pm$ 0.03
&& 100.0 & 4.22 $\pm$ 0.06
&& 100.0 & 5.10 $\pm$ 0.04
&&  92.1 & 4.51 $\pm$ 0.03\\

$10^5$
&   99.9 & 49.71 $\pm$ 0.33
&& 100.0 & 51.40 $\pm$ 0.06
&& 100.0 & 56.53 $\pm$ 0.61
&&  99.3 & 50.03 $\pm$ 0.20\\

\rowcolor{Gray}
$10^6$
&  --- & ---
&& --- & ---
&& --- & ---
&& --- & ---\\

\midrule
& & \multicolumn{9}{c}{\textbf{RASTER$'$ (Python)}}
\\
\rowcolor{Gray}
$10^2$
&   87.0 & 0.04 $\pm$ 0.00
&& 100.0 & 0.04 $\pm$ 0.00
&& 100.0 & 0.05 $\pm$ 0.00
&&  87.0 & 0.06 $\pm$ 0.00\\

$10^3$
&   77.6 & 0.44 $\pm$ 0.00
&& 100.0 & 0.49 $\pm$ 0.03
&& 100.0 & 0.81 $\pm$ 0.00
&&  93.7 & 0.84 $\pm$ 0.05\\

\rowcolor{Gray}
$10^4$
&   78.7 &  6.09 $\pm$ 0.03
&& 100.0 &  7.27 $\pm$ 0.51
&& 100.0 &  9.12 $\pm$ 0.30
&&  92.1 &  9.56 $\pm$ 0.04\\

$10^5$
&   99.9 &  74.45 $\pm$ 0.51
&& 100.0 &  84.74 $\pm$ 0.67
&& 100.0 & 107.64 $\pm$ 0.37
&&  99.3 & 121.79 $\pm$ 5.49\\

\rowcolor{Gray}
$10^6$
&  --- & --- 
&& --- & --- 
&& --- & --- 
&& --- & ---\\

\midrule

& & \multicolumn{9}{c}{\textbf{RASTER (Rust)}}
\\
\rowcolor{Gray}
$10^2$
&   87.0 & $<$ 0.01 $\pm$ 0.00
&& 100.0 & $<$ 0.01 $\pm$ 0.00
&& 100.0 & $<$ 0.01 $\pm$ 0.00
&&  87.0 & $<$ 0.01 $\pm$ 0.00\\

$10^3$
&   77.6 & 0.01 $\pm$ 0.00
&& 100.0 & 0.01 $\pm$ 0.00
&& 100.0 & 0.01 $\pm$ 0.00
&&  93.7 & 0.02 $\pm$ 0.00\\

\rowcolor{Gray}
$10^4$
&   78.7 & 0.08 $\pm$ 0.00
&& 100.0 & 0.15 $\pm$ 0.00
&& 100.0 & 0.28 $\pm$ 0.00
&&  92.1 & 0.40 $\pm$ 0.00\\

$10^5$
&   99.9 & 1.38 $\pm$ 0.00
&& 100.0 & 2.53 $\pm$ 0.00
&& 100.0 & 4.74 $\pm$ 0.01
&&  99.3 & 5.09 $\pm$ 0.01\\

\rowcolor{Gray}
$10^6$
&   99.9 & 33.85 $\pm$ 0.07
&& 100.0 & 48.59 $\pm$ 0.04
&& 100.0 & 60.61 $\pm$ 0.05
&&  99.3 & 58.05 $\pm$ 0.10\\

\midrule
& & \multicolumn{9}{c}{\textbf{RASTER$'$ (Rust)}}\\

\rowcolor{Gray}
$10^2$
&   87.0 & $<$ 0.01 $\pm$ 0.00
&& 100.0 & $<$ 0.01 $\pm$ 0.00
&& 100.0 & $<$ 0.01 $\pm$ 0.00
&&  87.0 &     0.01 $\pm$ 0.00\\

$10^3$
&   77.6 & 0.01 $\pm$ 0.00
&& 100.0 & 0.02 $\pm$ 0.00
&& 100.0 & 0.04 $\pm$ 0.00
&&  93.7 & 0.09 $\pm$ 0.00\\

\rowcolor{Gray}
$10^4
$&  78.7 & 0.24 $\pm$ 0.00
&& 100.0 & 0.50 $\pm$ 0.00
&& 100.0 & 0.95 $\pm$ 0.01
&&  92.1 & 1.62 $\pm$ 0.00\\

$10^5$
&   99.9 &  4.32 $\pm$ 0.15
&& 100.0 &  6.12 $\pm$ 0.31
&& 100.0 & 10.75 $\pm$ 0.48
&&  99.3 & 19.71 $\pm$ 0.03\\

\rowcolor{Gray}
$10^6$
&   99.9 &  64.46 $\pm$ 0.89
&& 100.0 &  82.92 $\pm$ 3.20
&& 100.0 & 140.59 $\pm$ 6.81
&&  99.3 & 248.33 $\pm$ 0.37\\
\bottomrule\end{tabular}\caption{Comparing (sequential) stand-alone implementations of RASTER and RASTER$'$ in Python 3.6 and Rust 1.35 at different precision values. The provided figures show the averages runtime $t$ of five runs, along with their standard deviation. The chosen parameters were  threshold $\tau = 5$ and minimum cluster size $\mu = 4$. The input size is stated as the number of clusters, where each cluster contains 500 points. The most noteworthy aspect is that fine-tuning the precision parameter entails performance improvement without deteriorating the quality of the results, with the caveat that the clusters identified by tiles in RASTER are coarser with a lower precision.}
\label{table:seq}
\end{table*}

\subsection{Data}
RASTER is designed for clustering unlabeled data where there is no ground truth available. Thus, it is an example of unsupervised learning. However, the input data for both experiments was produced by a custom data generator. Its output is synthetic data that is an idealization of an existing proprietary real-world data set. For such data sets, we know how many clusters the algorithms should detect, but the data is not labeled. For the first experiment, using the GPS coordinate plane as a canvas, the data generator randomly determines a predefined number of hypothetical cluster centers and spreads 500 points around them in a uniform distribution with random parameters, which means that clusters vary in their density. Cluster centers are located at least a given minimum distance apart. For the second experiment the data generator was modified to speed up file creation for larger files. This variant of our data generator divides the canvas into a number of rectangles and places cluster centers in them. This affects files with $10^5$ and $10^6$ cluster centers. The resulting distribution of cluster centers is less random than it is with our data sets that contain fewer clusters. As one idiosyncrasy of the generated data is that data points forming a cluster are produced in sequence, all data was randomly shuffled, based on the reasoning that partially sorted data may unfairly advantage some algorithms, while randomly shuffled data can be expected to not provide a discernible advantage to any algorithm.

Using a real-world data set instead of a synthetic data set would have been suboptimal for a number of reasons. First, there is the issue of confidentiality. The real-world data set that inspired the work described in this paper was given to us on the basis of strict confidentiality, governed by a non-disclosure agreement. Thus, we are legally not able to share this data. Second, the size of our real-world data set is far too large to make it easily shareable. In contrast, the synthetic data generator we share is defined as a relatively short Python program. Third, real-world data normally needs to be cleaned before it can be used. There is arguably little scientific value in this step. In our view, the synthetic data we provide, via a precise data definition, is superior to a real-world data set as it can be used as it is. It is trivial to generate large data sets this way. As our subsequent results show, even comparatively small data sets already push other standard clustering algorithms to their limits, so there is probably not much to be gained from processing larger data sets in this comparative performance evaluation.



\begin{figure*}
\centering
\begin{subfigure}{.20\textwidth}
  \centering
  \includegraphics[scale=0.45, trim=0.2cm 0.2cm 0.2cm 0.2cm, clip]{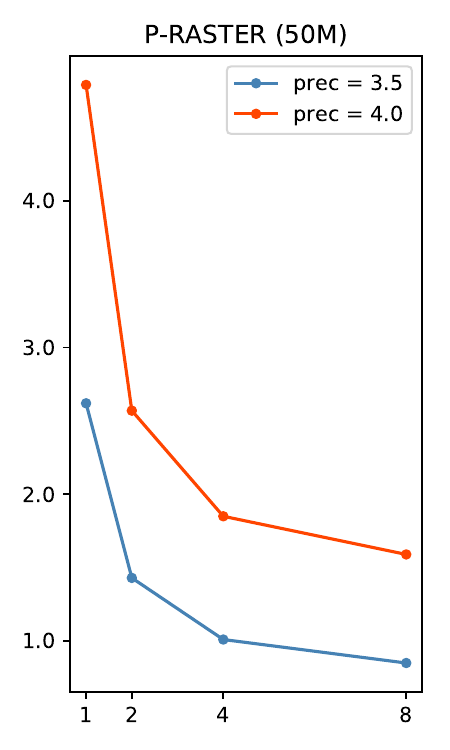}
  \label{fig:sub1}
\end{subfigure}\hfill
\begin{subfigure}{.20\textwidth}
  \centering
  \includegraphics[scale=0.45, trim=0.2cm 0.2cm 0.2cm 0.2cm, clip]{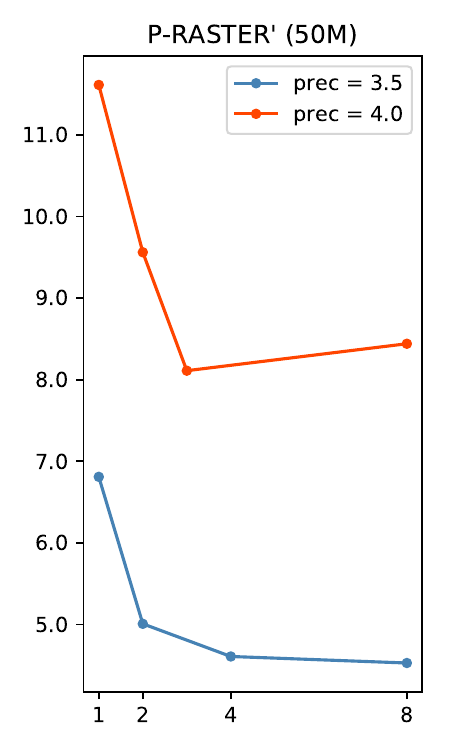}
  \label{fig:sub2}
\end{subfigure}\hfill
\begin{subfigure}{.20\textwidth}
  \centering
  \includegraphics[scale=0.45, trim=0.2cm 0.2cm 0.2cm 0.2cm, clip]{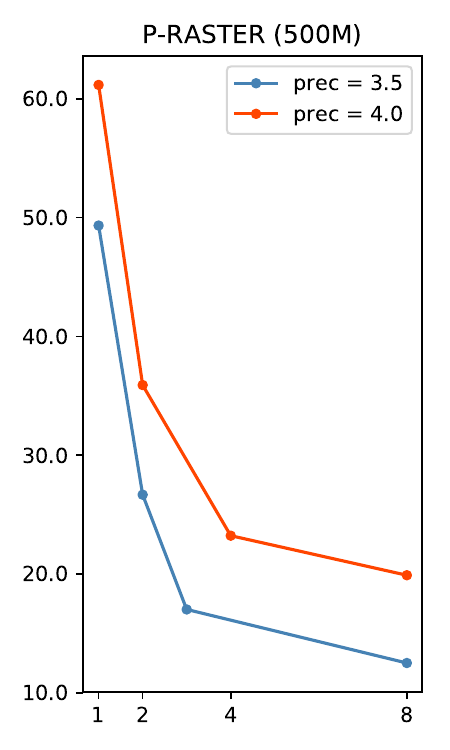}
  \label{fig:sub3}
\end{subfigure}\hfill
\begin{subfigure}{.20\textwidth}
  \centering
  \includegraphics[scale=0.45, trim=0.2cm 0.2cm 0.2cm 0.2cm, clip]{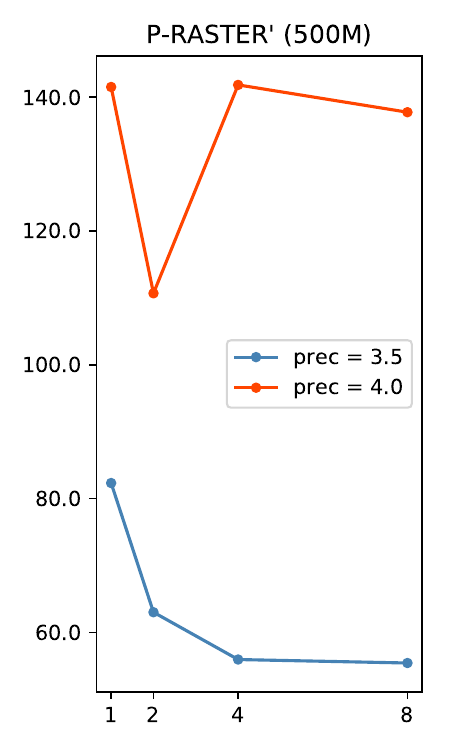}
  \label{fig:sub4}
\end{subfigure}
\vspace{-1.5em}
\caption{Scalability of P-RASTER and P-RASTER$'$ when clustering 50M and 500M points (best viewed in color). With $c$ cores, P-RASTER achieves a performance improvement of around $\frac{c}{2}$. The scalability of P-RASTER$'$, on the other hand, is limited, due to cache synchronization and parallel slowdown issues. The latter is, in particular, an issue with very large data sets.}
\label{fig:par}
\end{figure*}

\subsection{Results}
\label{raster-results}
\paragraph{Experiment 1}
Table~\ref{table:comparison} lists numerical results of ten clustering algorithms, relating to experiment 1. The stated averages are of ten runs with standard deviations in case of the runtime. Missing entries are due to memory constraints or excessive runtime. Some of these algorithms, such as Affinity Propagation and Spectral Clustering, were not even able to process 500K points in 64 GB RAM, even though the input corresponds to less than 2 MB. Overall, the hamstrung version of RASTER$\star$ clearly outperforms the other algorithms, with a runtime difference ranging from a single-digit multiple in the case of DBSCAN to a triple-digit multiple in the case of CLIQUE. Furthermore, Spectral Clustering did not honor setting the number of jobs to 1 and instead used up all available threads. Since this algorithm is a transformation of $k$-means, its runtime is strictly worse than the latter's. Gaussian Mixture Modelling also used all available threads at times. The instances of unwanted parallelization are probably due to \texttt{numpy} calling multi-threaded C-libraries.\footnote{Trying to artificially limit the number of used threads with the Linux tool \texttt{numactl} lead to erroneous behavior of the affected algorithms.}

We also looked at RAM consumption. The maximum input size the various algorithms were able to process is as follows: 5K points ($10^1$ clusters) in the case of Affinity Propagation and Spectral Clustering; 50K points ($10^2$ clusters) in the case of Ward's method, Agglomerative Clustering and CLIQUE; $10^3$ clusters (500K points) in the case of Mini-batch $k$-means, BIRCH and Gaussian Mixture; 5M points ($10^4$ clusters) in the case of DBSCAN; 50M points ($10^5$ clusters) in the case of RASTER$\star$ and Mean Shift. Mini-batch $k$-means had an overly excessive runtime, projected to be multiple hours, with 5M points ($10^4$ clusters); CLIQUE showed the same problem with 500K points ($10^3$ clusters) already. Thus, we terminated these benchmarks early. Even though Mean-Shift and RASTER$\star$ would be able to process larger inputs, the design of that particular benchmark was a limiting factor as all data is stored in memory and retained. Regarding Mini-batch $k$-means, it is also worth pointing out that the algorithm occasionally seems to get trapped in a local optimum, which leads to runtimes that are over 10 times worse. The implementation in \texttt{scikit-learn} terminates those eventually. In those rare outliers, both the silhouette coefficient and the percentage of identified clusters drop significantly.

\paragraph{Experiment 2}
In Table~\ref{table:seq}, a better overview of the true performance of RASTER and RASTER$'$ is given. The Python implementations were not able to process files with 500M input points (1M clusters) as they ran out of memory.\footnote{This could be fixed by modifying the implementation to not use relatively memory-intensive tuples for storing GPS coordinates. We were more interested in a straightforward implementation of RASTER in Python, however, which is why this was not done.} In general, RASTER scales well as the input size increases. The chosen precision value affects both runtime and quality of results. For example, a precision of 3 is fast, yet the resulting coarse grid detects fewer clusters.

\paragraph{Experiment 3}
As the results in Fig.~\ref{fig:par} show, both RASTER and RASTER$'$ scale well. With RASTER, the total speedup with 8 cores is between 3 and 4. In contrast, the performance improvements with RASTER$'$ are greatly affected by the chosen precision. With a precision of 3.5, performance improvements are above 30\% with four cores, but level off with 8 cores. On the other hand, with a precision of 4, performance only improves when using up to 4 course with the 50M data set. With the 500M data set, performance improves by ca.~22\% with two cores, but deteriorates with four and eight cores. A detailed breakdown of the runtime of the projection and clustering steps of P-RASTER and P-RASTER$'$ is provided in \ref{app-d}.





\paragraph{Discussion}
In experiment 1, RASTER$\star$ outperforms an assortment of other clustering algorithms by a wide margin. Given the popularity of \texttt{scikit-learn}, these implementations can be considered state-of-the-art. Of course, those algorithms may have been developed for other use cases than ours (cf.~\ref{app-a} for a brief discussion of RASTER as a general-purpose density-based clustering algorithm). Not only is their runtime worse, many of them have memory requirements that make them unsuitable for density-based clustering of big data. Even though RASTER$\star$ was hamstrung in that comparison, it still has the best performance, and by a wide margin. Furthermore, as a comparison of the data in Table ~\ref{table:comparison} and Table \ref{table:seq} shows, RASTER is about five times faster than RASTER$\star$, without those limitations. 

There is a conceptual similarity between a sub-method of CLIQUE and RASTER, but the former performs lookup in a way that scales poorly as the input increases. This may not be obvious when using inputs with very few clusters, as it was done in the original papers~\cite{Agrawal1998, Agrawal2005}. Yet, CLIQUE is not an ideal choice in a big data context. We spent a considerable amount of time on comparing RASTER and CLIQUE. There are very few widely available implementations of the latter, however, which perhaps underscores that CLIQUE is more interesting from a theoretical perspective. We consulted the implementations of the \texttt{pyclustering} package\footnote{The code repository of the package \texttt{pyclustering} is \url{https://github.com/annoviko/pyclustering} (accessed 15 April 2019).} as well as Gy{\"o}rgy Katona's implementation. Our implementation is based on the latter and we improved its performance by a factor of more than 3. It is also noticeably faster than the implementation in \texttt{pyclustering}. Our results comparing RASTER to another widely available implementation of CLIQUE in Java are documented in \ref{app-b}. Thus, it seems that the comparatively low performance of CLIQUE is not due to a substandard implementation.

In experiment 2, we have shown that RASTER and RASTER$'$ scale very well. The purpose of the additional implementation of RASTER in Rust is to show our algorithm's potential for real-world workloads. This implementation is able to process 500M points (1M clusters) in less than 50 seconds with a single thread. Unlike many other clustering algorithms, RASTER can be easily parallelized. The provided data also illustrate that there is an ideal range of values for the precision parameter, which is a measure of the resolution of the implied grid. With a precision value of 3, there is some misclassification because, for instance, two clusters in the data may be classified as one. At the other extreme, with a precision value of 5, a grid that is too fine may also lead to a reduced precision, but in that case the reason is that the elements of a cluster may be mapped to tiles that are too far apart for the used distance value. Subsequently, these scattered tiles do not meet the chosen minimum cluster size $\mu$ and are filtered out prior to the agglomeration step.

As our results of experiment 3 show, RASTER can not only be elegantly parallelized in theory; the measured performance improvements also show a substantial performance gain. Based on these results, we would expect that a further increase in the number of available CPU cores $c$ leads, for P-RASTER, to a speedup of roughly $\frac{c}{2}$ for the projection step and $c$ for the clustering step. The improvement for clustering in P-RASTER$'$ is similar, but the performance improvement for projection seems to level off. As the current picture is ambiguous, it is hard to make predictions on expected performance increases of that part of the algorithm.

We have found that with a greater precision value in P-RASTER$'$, total performance not only levels off, as it was the case with a precision value of 3, but degrades, with a precision of 4. This seems to be due to memory allocation for the data that is retained. In addition to increased memory requirements, cache synchronization issues may also play a role. Since each core on a multi-core CPU has its own cache, there is a greater need for cache synchronization. This phenomenon is referred to as parallel slowdown.

We also took a look at related work on parallelizing other clustering algorithms. The state of the art in parallelizing DBSCAN is Song et al.'s work on RP-DBSCAN, which they describe as "superfast"~\cite{song2018rp}. Yet, the achieved speedup with 10 cores is only around 2, reaching 4.4 with 40 cores. In comparison, P-RASTER achieves a parallel speedup of a factor of around 4 already with 8 cores, due to the excellent scalability inherent in the design of our algorithm. Based on our experiments, we would expect a speedup of around 20 with 40 cores. By design, RASTER can be parallelized rather effectively. Yet, parallelizing many other clustering algorithms, such as DBSCAN, is much less straightforward. Such efforts may also have a relatively low ceiling, as suggested by the empirical performance of RP-DBSCAN.

\section{Related Work}
\label{sec:related}
We compared RASTER to a number of canonical clustering algorithms of which implementations were readily available. There are other algorithms, which seem to have been of more academic interest. For reasons related to space as well as a lack of availability of implementations, we therefore briefly discuss other density-based clustering algorithms separately below. In general, even though we point out some superficial similarities between RASTER and various other algorithms, we would like to stress that none of those algorithms, as described in the respective papers, seems suitable for big data clustering.

An early approach to grid-based spatial data mining was STING~\cite{wang1997sting}. A key difference between RASTER and this algorithm is that it performs statistical queries, using distributions of attribute values. WaveCluster~\cite{sheikholeslami1998wavecluster} shares some similarities with RASTER. It can reduce the resolution of the input, which leads to output that is visually similar to RASTER clusters that are defined by its significant tiles. WaveCluster runs in linear time. However, because the computation of wavelets is costlier than the operations RASTER performs, its empirical runtime is presumably worse.

The projection step of RASTER makes use of a standard concept in data processing: assigning values to buckets. Similar corresponding projections are a feature in other algorithms as well. For instance, Baker and Valleron~\cite{baker2014open} present a solution to a problem in spatial epidemiology whose initial step seems similar to the projection step performed by RASTER. For solving their problem, it is sufficient to count data points in squares of a grid. The approach taken by GRPDBSCAN, discovered by Darong and Peng~\cite{darong2012grid} is similar to the aforementioned work in this regard, but chooses the approach of RASTER$'$ by retaining the projected input points. Unfortunately, their description seems rather vague. Because it includes only an illustration instead of an implementation or specification, we cannot evaluate how similar their algorithm is to ours. In both of those papers there is no discussion of using a variable scaling factor, however, to arbitrarily adjust the granularity of the grid and the resulting clusters. Lastly, there is some similarity between our algorithm, when focusing only on the hub identification problem, and blob detection in image analysis~\cite{danker1981blob}. A direct application of that method to finding clusters would arguably require very dense cluster centers, which could be achieved by projecting points to fewer tiles. Also, blob detection is computationally more expensive than RASTER. Furthermore, unlike blob-detection algorithms, RASTER is viable as a general-purpose clustering algorithm, as illustrated in~\ref{app-a}.

\section{Future Work}
\label{sec:future}


Xiaoyun et al.~introduced GMDBSCAN~\cite{xiaoyun2008gmdbscan}, a DBSCAN-variant that is able to detect clusters of different densities. The inability to detect such clusters is shared between DBSCAN and RASTER. For more general purpose-applications, it may be worth investigating a similar approach for RASTER as well. A starting point is an adaptive distance parameter for significant tiles. While we mentioned that the distance metric $\delta$ can be arbitrarily chosen --- we picked the eight immediate neighbors of a tile --- one could not only choose different Manhattan or Chebyshev distances, but arbitrary adaptive values. On the other hand, by fine-tuning the precision parameter value, very similar results could be achieved, possibly with multiple passes over the same input, which would make it possible to detect clusters of varying densities.

RASTER does not distinguish between significant tiles. Yet, the number of points projected to different significant tiles can differ widely. Some tiles could have a very high count of points, while others barely reach the specified threshold value. One obvious modification would be to use those counts for visualization purposes, for instance by rendering tiles with fewer points in a lighter color tone and those with more points in a darker one. However, there are more sophisticated approaches for utilizing relative densities of tiles. Thus, one could consider an adaptive approach to \textsc{RASTER} clustering, for instance by subdividing such tiles into smaller segments, with the goal of determining more accurate cluster shapes. This idea is related to adaptive mesh refinement, suggested by Liao et.\ al~\cite{liao2004grid}. A related idea is to change the behavior of \textsc{RASTER} when detecting a large number of adjacent tiles that have not been classified as significant. This may prompt a coarsening of the grid size for that part of the input space.

For practical use, it may be worthwhile to add a contextual relaxation value $\epsilon$ for the threshold value of significant tiles. For instance, in the vicinity of several significant tiles, a neighboring tile with $\tau - \epsilon$ data points may be considered part of the agglomeration, in particular if it has multiple significant tiles as neighbors. A key aspect of RASTER is that it returns clusters of significant tiles that cover most of the area of a dense cluster. It may be interesting to not return such clusters at all but instead compute the ellipse of the least area that includes all tiles of a cluster. In particular for larger clusters, this would lead to an even more memory-efficient output. This also has a fitting correspondence to RASTER$'$ as the ellipse could be computed based on all the points projected to the tiles in the final set of clusters.

RASTER is highly suited to clustering data batches, but we have also implemented a variant of RASTER for data streams (S-RASTER). In an upcoming paper, we intend to discuss this variant and show how it compares with other algorithms for clustering evolving data streams~\cite{ulm2019s}. Lastly, P-RASTER was discussed as an algorithm running on a many-core CPU. It should be straightforward to adopt it to execution on the cloud, similar to the experiments described in the RP-DBSCAN paper~\cite{song2018rp}. We would expect it to scale well as slices can be processed in parallel. Unlike RP-DBSCAN, P-RASTER on the cloud would not need to duplicate any data, implying better scalability in terms of required memory as well. That being said, P-RASTER is already able to process terabytes of data in a very reasonable amount of time on a single workstation. This means that adapting P-RASTER to be executed on a data center like Microsoft Azure, Amazon Web Services, or Google Cloud would be hard to justify in the foreseeable future.

Lastly, while a detailed evaluation of the parallel RASTER variant P-RASTER was not the primary focus of this paper, it could be interesting to also carry out comparisons to other parallel clustering algorithms. As our evaluation shows, sequential RASTER is already highly performant. We have also shown that P-RASTER scales well on multicore machines. Given that the RASTER family of  algorithms is able to process terabytes of data, which very comfortably covers the practical use case for which we developed it, a detailed comparison with parallel clustering algorithms would arguably be of limited utility. Another hindrance to such a project is that there is no readily available suite of implementations of such algorithms.


\section{Conclusion}
\label{sec:conclusion}
We hope to have shown that RASTER is an excellent density-based clustering algorithm with an outstanding single-threaded performance. For our particular problem, it outperforms standard implementations of existing clustering algorithms. On top, unlike many other clustering algorithms, RASTER can be effectively parallelized to make use of many-core CPUs. Based on our experiments, we expect RASTER to scale linearly with the number of available CPU cores. As the input space of RASTER can be partitioned without data duplication and without affecting clustering quality, it would be straightforward to adopt our algorithm for cloud-scale computations, albeit that may not be necessary, considering the performance that can be achieved on a regular workstation.

While we focused on the hub identification problem, RASTER can also be used for solving general-purpose density-based clustering problems. The parameters RASTER uses, i.e. the threshold value $\tau$ for significant tiles, the distance metric $\delta$, and the minimum cluster size $\mu$, are all intuitive. In~\ref{app-a}, we illustrate on standard data sets how our algorithm can deliver very good results with minimal parameter tweaking. RASTER is particularly suited to low-dimensional data. Overall, we have shown that RASTER is of great practical value due to its very fast sequential and parallel performance and its suitability for general-purpose clustering problems.

\section*{Acknowledgments}
This research was supported by the project \emph{Fleet telematics big data analytics for vehicle usage modeling and analysis (FUMA)} in the funding program \emph{FFI: Strategic Vehicle Research and Innovation (DNR 2016-02207)}, which is administered by VINNOVA, the Swedish Government Agency for Innovation Systems. Initial multi-core benchmarks were performed on resources at the Chalmers Centre for Computational Science and Engineering (C3SE), provided by the Swedish National Infrastructure for Computing (SNIC).




%
%


\bibliographystyle{elsarticle-num} 
\bibliography{RASTER-refs}






\appendix
\section{Qualitative Comparison of Clustering Algorithms}
\label{app-a}
Arguably the most well-known qualitative comparison of prominent clustering algorithms is provided by the online documentation of \texttt{scikit-learn}.\footnote{Refer to the Section "Comparing different clustering algorithms on toy data sets" in the official documentation of \texttt{scikit-learn}, which is available at
\url{https://scikit-learn.org/0.20/auto_examples/cluster/plot_cluster_comparison.html} (Accessed April 4, 2019).} We have modified RASTER$'$ to work with that benchmark. The results are provided in Figure~\ref{fig:comparison}. The RASTER parameters we used were threshold $\tau$ = 5, minimum cluster size $\mu$ = 5, and a precision value of 0.9. The runtime metrics have to be taken with a grain of salt as our algorithm has to process the input twice to produce the desired labels. This is not an inherent flaw of RASTER but instead due to the implementation of that particular comparison. Furthermore, our implementation has not been optimized. For instance, we convert from Python to \texttt{numpy} objects and back. The standalone implementation of RASTER$'$ furthermore does not require us to return a label for each element in the provided input, in the exact same order. It does not retain all data either as all noise is simply discarded. In the provided examples, the input size is a mere 1,500 points each, with very few clusters. As we have shown earlier (cf.~Sect.~\ref{raster-results}), our algorithm shines with huge data sets and very large numbers of clusters.

\begin{sidewaysfigure*}
\centering
\includegraphics[scale=0.55, trim=0.0cm 0.0cm 0.0cm 0.4cm, clip]{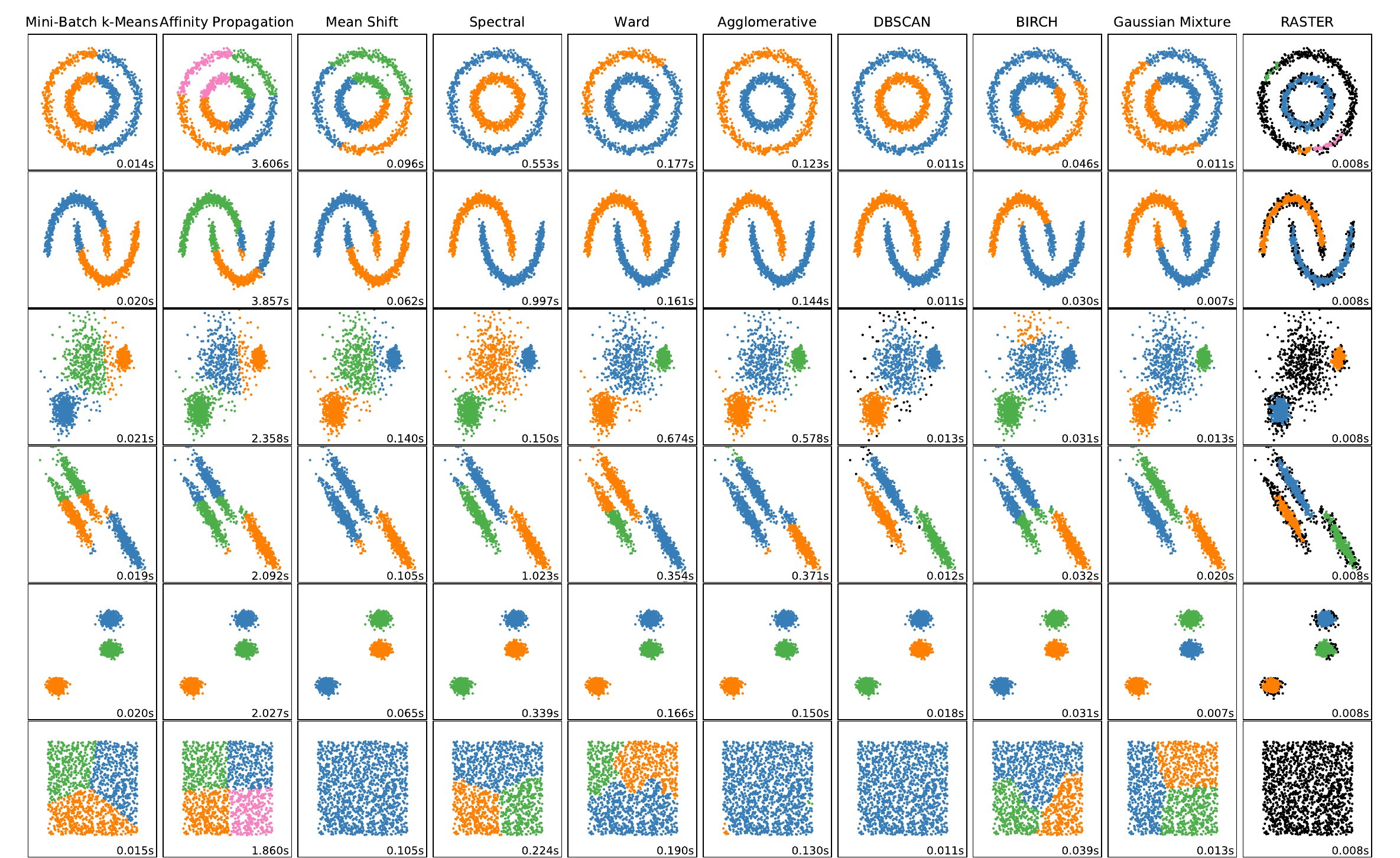}
\caption{This comparison (best viewed in color) between RASTER$\star$ and various standard clustering algorithm shows that our algorithm approximately cluster centers of tight, dense clusters very well. The third data set from the top reflects the hub identification problem where tight, dense clusters have to be found in noisy data sets. RASTER is the only algorithm in this comparison that adequately does this. Furthermore, it is the only algorithm in this comparison that does not produce any clusters with homogeneous input, as seen in the bottom image. Furthermore, the output of the first data set can be easily improved by modifying the precision parameter (cf.~Fig.~\ref{fig:rvsr}).}
\label{fig:comparison}
\end{sidewaysfigure*}

Having pointed out that this benchmark is not an ideal match for RASTER, we can nonetheless confirm that the results are in line with our expectations as our algorithm turned out to perform very well also with smaller input sizes and fewer clusters. To go through the six input data sets from top to bottom: The concentric circles in the first data set are not suited for the chosen RASTER parameters, for which the outer circle is not dense enough. The inner one, however, is properly identified. By choosing a different precision value, the outer ring can be clearly identified (cf.~Fig.~\ref{fig:rvsr}), but this would negatively affect the results of the other data sets. The arcs in the second data set are nicely separated, showing that RASTER does not require dense, circular clusters. The third one contains two denser clusters on the left and right, while the middle is less dense. RASTER is the only algorithm that considers the middle as noise. There is no ground truth provided. Yet, that is how we would interpret the GPS data we have worked with as well, which means that this data set is relevant to the hub identification problem, where tight clusters are surrounded by noise. As our example shows, RASTER ignores the less densely spread points between two dense clusters, while the other algorithms all detect three clusters. This seems to imply that, if speed and memory were not a concern, standard clustering algorithms still could not be used for the hub identification problem without preprocessing the data and filtering out noise.

The fourth and particularly the fifth data set are ideally suited for RASTER, which is shown to reliably identify hubs. Lastly, the sixth data set is worth highlighting. According to the information provided in the \texttt{scikit-learn} documentation, such homogeneous data is a "null situation for clustering", stating that there is "no good clustering" possible. Interestingly, RASTER is the only of the ten clustering algorithms in this comparison that does not perform any clustering on this data set. Instead, all data is discarded as noise. It is also noteworthy that the indicated runtime of RASTER does not depend on the number of clusters or the points per cluster. Instead, it is primarily dependent on the number of points in the input. As the input size is constant in the various data sets of this example, the resulting runtime is identical. This is in stark contrast to the other algorithms, which have sometimes widely fluctuating runtimes.

As stated earlier, adjusting the precision parameter can help improve clustering results. This is shown in Fig.~\ref{fig:rvsr}, which contrasts the output of RASTER with a precision value of 0.90 with the results of an execution with a precision of 0.73. As there is no ground truth in unsupervised learning, it is in the eye of the beholder which results are judged to be satisfactory. Yet, as this juxtaposition shows, with minor adjustments very good results can be achieved with each of the used reference data sets. We would argue that either precision value delivers very good results in data sets 2, 4, and 5. Using a precision value of 0.73 seems preferable for data set 1 as it fully clusters both rings. For data set 3, an argument could be made for either of those two parameter values, depending on whether the less densely scattered points in the middle are desired to be clustered or not. Lastly, a precision value of 0.90 delivers superior results for data set 6 as it does not cluster a homogeneous data set. In \ref{app-robustness}, the impact of changes of hyperparameter values will be discussed in more detail.

\begin{figure}
\centering
\includegraphics[scale=0.60, trim=0.0cm 0.0cm 0.0cm 0.7cm, clip]{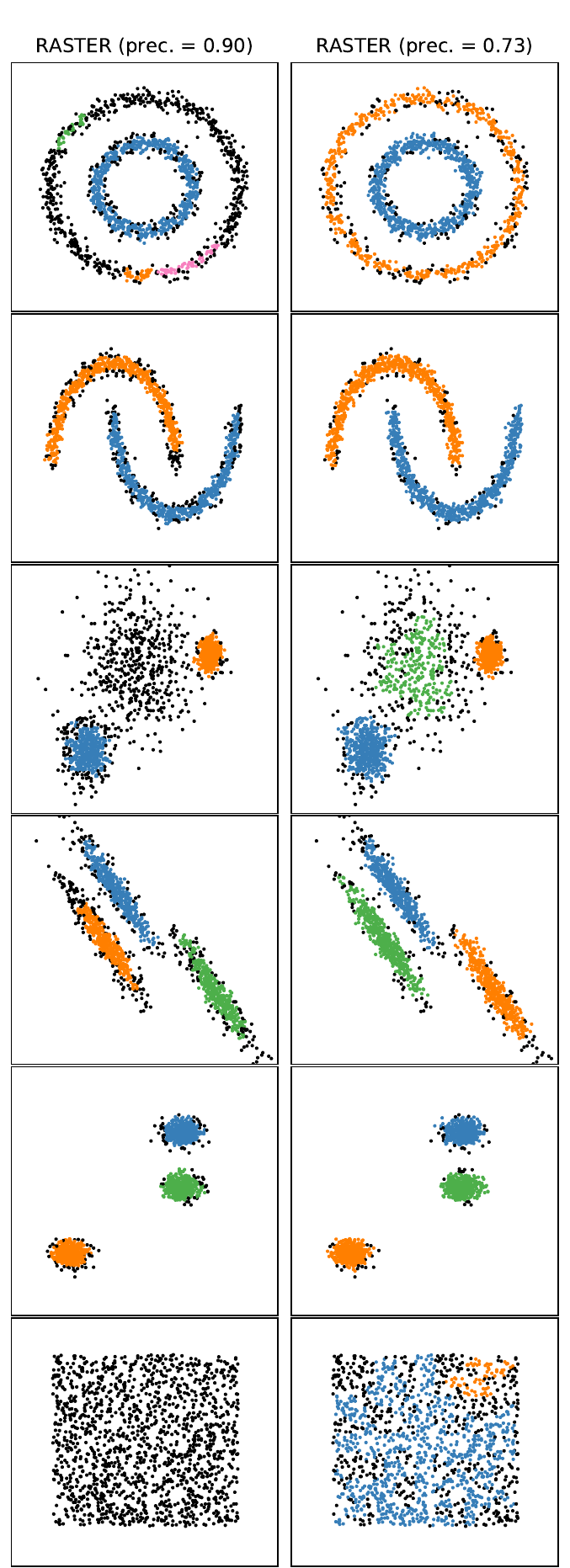}
\caption{Adjusting the precision parameter can improve clustering results (best viewed in color). With a precision of 0.90, all but the very first data set are clustered satisfactorily. Reducing the precision to 0.73 improves the results of the first data set. It is a matter of debate whether the third data set is better clustered with a precision of 0.73. Also, with a reduced precision, RASTER attempts to cluster the homogeneous data set, while it is noteworthy (cf.~Fig.~\ref{fig:comparison}) that RASTER, with a precision of 0.90, is the only algorithm in that comparison that does not cluster this data set.}
\label{fig:rvsr}
\end{figure}

\section{Clustering Quality and Hyperparameter Values}
\label{app-robustness}
RASTER takes into account various hyperparameters. Referring to our specification in Alg.~\ref{alg:raster}, these are precision \emph{prec}, threshold~$\tau$, distance~$\delta$, and minimum cluster size~$\mu$. The distance parameter is not exposed in the provided implementation, which only considers directly adjacent tiles. This approach was sufficient for the experiments described in Sect.~\ref{raster-eval}. Suitable hyperparameter values need to be found by visual inspection, as is also the case with the other clustering algorithms we benchmarked in the performance comparison in Sect.~\ref{raster-eval} as well as the qualitative comparative analysis in \ref{app-a}. In Fig.~\ref{fig:robustness} the effect of successive hyperparameter modification is illustrated. The biggest immediate effect is due to the precision parameter. Thus, the first four columns show a step-wise change until a value is found that shows good results. In contrast, in Fig.~\ref{fig:rvsr}, which is discussed in \ref{app-a}, the immediate steps are omitted.

Once a suitable hyperparameter value for the precision input has been found, further adjustments can be made. In the fifth and sixth column of Fig.~\ref{fig:robustness} the threshold parameter value was increased, which arguably did not improve the results. Alternatively, the seventh and eighth column of the same figure increase the value for the minimum cluster size. In the seventh column, this makes it possible to exclude the outer circle of the first image, which may or may not be practically relevant. In the eighth column, the third image shows that increasing the minimum cluster size leads to the exclusion of the sparse collection of points in the middle. Again, depending on the use case, this may be an improvement. Lastly, the ninth and tenth columns show the results of increasing the threshold and minimum-cluster-size parameters simultaneously. As these images show, as would be expected, the number of identified clusters declines.

The ten examples in Fig.~\ref{fig:robustness} serve as an illustration. Given the number of hyperparameters, their possible values, and the potentially limitless variety of the input data, a more extensive exploration would quickly enter the realm of a combinatorial explosion. Nonetheless, our approach should give a good idea of the applicability of RASTER to a variety of input data as well as show how its hyperparameter values can be effectively adjusted in practice. In many cases, similar to ours, it may only be necessary to adjust the precision value.

\section{Performance Comparison between CLIQUE, RASTER, and RASTER$'$}
\label{app-b}
As there is a superficial similarity between one step of CLIQUE and the neighborhood-lookup step of RASTER, we were particularly interested in comparing the performance of these two algorithms on the hub identification problem. In Sect.~\ref{raster-results} we show how a Python implementation of CLIQUE compares against RASTER. However, we also came across a Java implementation of CLIQUE as part of the OpenSubspace library~\cite{muller2009evaluating}.\footnote{The repository associated with that paper is no longer available online. However, the relevant Java packages were retrieved from version 1.0.4 of the R library \texttt{subspace}, which makes use of those Java libraries.} In order to perform a comparative benchmark, we implemented RASTER and RASTER$'$ in Java and pitted them against that implementation of CLIQUE. In Table~\ref{tab:clique} we show the results, which illustrate that the comparatively substandard performance of CLIQUE was not just due to the Python implementations we used earlier. In Java, the difference is likewise rather substantial. Further, the bigger the input, the wider the gap between CLIQUE and RASTER gets. With an input of 50K points or 100 clusters, RASTER is 300 times faster than CLIQUE. Increasing the input size to 500K points or 1K cluster widens the performance difference to well over 3,000 times.

\begin{table}[]
\centering
\begin{tabular}{@{}lrrr@{}}
\toprule
Size & \textbf{CLIQUE} & \textbf{RASTER} & \textbf{RASTER$'$} \\
\midrule
$10^1$ &   0.012 & $<$ 0.001 & $<$ 0.001 \\
$10^2$ &   1.498 &     0.005 &     0.008 \\
$10^3$ & 202.187 &     0.065 &     0.084 \\
\bottomrule
\end{tabular}
\caption{Performance comparison of CLIQUE, RASTER, and RASTER$'$ in Java. Runtimes are given in seconds and show how the various algorithms scale with increasing input sizes. For CLIQUE, we used a value of xi of 20, for RASTER, a precision of 3.5 was used. Not shown are the percentages of clusters identified. With the exception of CLIQUE with 500K points ($10^3$ clusters) at a 99.6 \% success rate, the algorithms detect all clusters.}
\label{tab:clique}
\end{table}

\begin{sidewaysfigure*}
\centering
\includegraphics[scale=0.41, trim=0.0cm 0.0cm 0.0cm 0.4cm, clip]{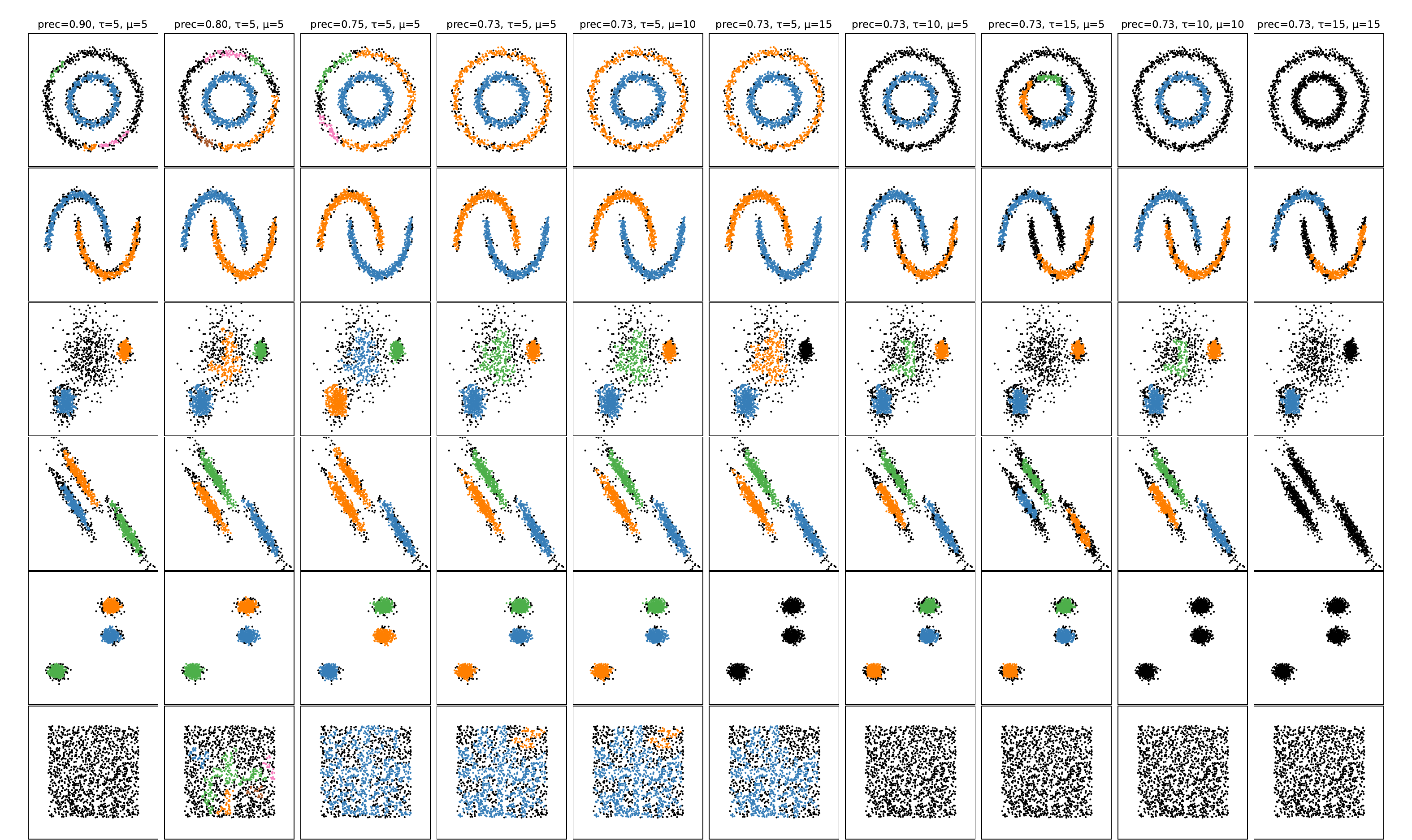}
\caption{This comparison (best viewed in color) of various applications of RASTER on varied datasets with changing hyperparameter values illustrates that clustering quality can vary significantly. The first four columns illustrate how adjusting the precision parameter leads to, arguably, increasingly better results. The determined ideal precision value of $0.73$ is held constant in columns five to ten. Columns five and six illustrate the effect of increasing the minimum cluster size while columns seven and eight show the changed output due to an increase in the minimum threshold value. Lastly, columns nine and ten show the changes that result from increasing both the minimum cluster size and the threshold value.}
\label{fig:robustness}
\end{sidewaysfigure*}

\section{Correctness Proofs for P-RASTER}
\label{app-c}
Below, we provide correctness proofs of three properties of P-RASTER that are true after the termination of the algorithm. These proofs complement Sec.~\ref{sec:p-raster}, including Alg.~\ref{alg:concurrent-clustering} and  Alg.~\ref{alg:join}. The three properties are as follows: (1) Clusters that have a neighboring significant tile in a different slice will have been joined, regardless of how often that border is crossed (\emph{multiple single-border crossings invariant}), (2) clusters having neighboring significant tiles in multiple slices will have been joined regardless how many borders they cross (\emph{consecutive border crossings invariant}), and (3) slicing can be arbitrary and does not affect the clusters in any way (\emph{arbitrary slicing invariant}).

There could potentially be a large number $O_c$ of clusters around a border which form one large cluster when joined. Figure~\ref{fig:dc1} illustrates a scenario where $O_c=3$ in the case of clusters $c_4$, $c_5$, and $c_8$. The second problem is that there could be a cluster within a slice that spans the length of the slice itself. This cluster can be joined from both its left side and its right side as can be seen in Fig.~\ref{fig:dc1} with $c_3$, $c_7$, and $c_9$. When $c_7$ is joined with the border to its right side, the new cluster needs to be carried over to the joining process around the border to its left. This leads to the two theorems below.

\newtheorem{lem}{Theorem}

\begin{lem}
Multiple single-border crossings invariant: Clusters that have neighboring significant tiles in an adjacent slice will be joined, regardless of how many neighboring significant tiles there are across that border and on which side of it they are located.
\end{lem}
 
\begin{proof}
Consider arbitrarily many clusters that are separated by the same border that would form one single cluster without that border. This implies that each cluster has at least one other cluster it is adjacent to in the neighboring slice. One of the original clusters will eventually be added to the set \textit{to\_visit}, processed, and have its neighbors determined. Each cluster is removed from that set as it is processed and its neighbors are added to it if they have not been added before. This procedure continues until there are no clusters remaining in \textit{to\_visit}. Thus, this part of the algorithm is an example of depth-first search (DFS)~\cite{even2011graph} where the searched nodes form a cluster.
\end{proof}

\begin{lem}
Consecutive border crossings invariant: Clusters having neighboring significant tiles in multiple slices will be joined regardless of how many borders they may cross.
\end{lem}

\begin{table*}[t]

\centering\ra{1.0}\begin{tabular}{@{}rrrrrrrrrrrr@{}}

\toprule
&
\multicolumn{3}{c}{prec = 3.5} & \phantom{i}&
\multicolumn{3}{c}{prec = 4}   & \phantom{i}
\\
\cmidrule{2-4}  \cmidrule{6-8}

cores 
& total & $\pi$ & $\kappa$ 
&& total & $\pi$ & $\kappa$ &

\\ \midrule

&  &
\multicolumn{5}{c}{\textbf{P-RASTER} / 50M points ($10^5$ clusters)}
\\
$1$
&  2.62 $\pm$ 0.00 & 2.05 $\pm$ 0.00 & 0.58 $\pm$ 0.00
&& 4.79 $\pm$ 0.03 & 3.67 $\pm$ 0.03 & 1.12 $\pm$ 0.01\\

$2$
&  1.43 $\pm$ 0.01 & 1.13 $\pm$ 0.01 & 0.30 $\pm$ 0.00
&& 2.57 $\pm$ 0.01 & 2.00 $\pm$ 0.01 & 0.57 $\pm$ 0.00
\\

$4$
&  1.01 $\pm$ 0.01 & 0.87 $\pm$ 0.00 & 0.14 $\pm$ 0.00
&& 1.85 $\pm$ 0.02 & 1.58 $\pm$ 0.02 & 0.27 $\pm$ 0.00
\\

$8$
&  0.85 $\pm$ 0.01 & 0.75 $\pm$ 0.01 & 0.09 $\pm$ 0.00
&& 1.59 $\pm$ 0.01 & 1.43 $\pm$ 0.01 & 0.16 $\pm$ 0.00
\\

&  &
\multicolumn{5}{c}{\textbf{P-RASTER} / 500M points ($10^6$ clusters)}
\\

$1$
&   49.33 $\pm$ 0.15 & 42.19 $\pm$ 0.14 &  7.14 $\pm$ 0.02
&&  61.17 $\pm$ 0.27 & 48.56 $\pm$ 0.23 & 12.61 $\pm$ 0.10
\\

$2$
&   26.67 $\pm$ 0.05 & 23.00 $\pm$ 0.06 & 3.66 $\pm$ 0.06
&&  35.90 $\pm$ 0.07 & 29.39 $\pm$ 0.12 & 6.52 $\pm$ 0.10
\\

$4$
&   17.01 $\pm$ 0.08 & 14.93 $\pm$ 0.07 & 2.08 $\pm$ 0.02
&&  23.22 $\pm$ 0.14 & 19.44 $\pm$ 0.12 & 3.78 $\pm$ 0.05
\\

$8$
&   12.50 $\pm$ 0.02 & 10.88 $\pm$ 0.03 & 1.61 $\pm$ 0.01
&&  19.89 $\pm$ 0.17 & 16.91 $\pm$ 0.15 & 2.98 $\pm$ 0.02
\\

\midrule

&  &
\multicolumn{5}{c}{\textbf{P-RASTER$'$} / 50M points ($10^5$ clusters)}
\\

$1$
&   6.81 $\pm$ 0.43 & 6.01 $\pm$ 0.43 & 0.80 $\pm$ 0.00
&& 11.61 $\pm$ \hphantom{0}0.98 & 9.83 $\pm$ \hphantom{0}0.93 & 1.78 $\pm$ 0.08
\\

$2$
&   5.01 $\pm$ 0.16 & 4.68 $\pm$ 0.16 & 0.33 $\pm$ 0.01
&&  9.56 $\pm$ \hphantom{0}0.40 & 8.93 $\pm$ \hphantom{0}0.40 & 0.63 $\pm$ 0.00
\\

$4$
&   4.61 $\pm$ 0.08 & 4.45 $\pm$ 0.08 & 0.16 $\pm$ 0.00
&&  8.11 $\pm$ \hphantom{0}0.25 & 7.80 $\pm$ \hphantom{0}0.25 & 0.31 $\pm$ 0.01
\\

$8$
&   4.53 $\pm$ 0.10 & 4.43 $\pm$ 0.09 & 0.10 $\pm$ 0.01
&&  8.44 $\pm$ \hphantom{0}0.22 & 8.24 $\pm$ \hphantom{0}0.22 & 0.19 $\pm$ 0.01
\\

&  &
\multicolumn{5}{c}{\textbf{P-RASTER$'$} / 500M points ($10^6$ clusters)}
\\
$1$
&   82.32 $\pm$ 4.20 &  72.73 $\pm$ 4.20 &  9.58 $\pm$ 0.01
&& 141.52 $\pm$ \hphantom{0}2.79 & 118.99 $\pm$ \hphantom{0}2.80 & 22.53 $\pm$ 0.77
\\

$2$
&   63.01 $\pm$ 2.23 &  59.03 $\pm$ 2.25 & 3.98 $\pm$ 0.02
&& 110.68 $\pm$ \hphantom{0}4.31 & 103.41 $\pm$ \hphantom{0}4.36 & 7.28 $\pm$ 0.06
\\

$4$
&   55.94 $\pm$ 2.14 &  53.70 $\pm$ 2.15 & 2.24 $\pm$ 0.01
&& 141.83 $\pm$ \hphantom{0}4.82 & 137.64 $\pm$ \hphantom{0}4.77 & 4.19 $\pm$ 0.05
\\

$8$
&   55.42 $\pm$  1.21 &  53.70 $\pm$  1.22 & 1.72 $\pm$ 0.01
&& 137.75 $\pm$ 11.32 & 134.43 $\pm$ 11.30 & 3.32 $\pm$ 0.07
\\

\bottomrule\end{tabular}\caption{Runtime of parallel implementations of RASTER and RASTER$'$ in Rust at precision levels 3.5 and 4 (5 runs each). The chosen parameter values for the threshold $\tau$ and minimum cluster size $\mu$ were 5 and 4, respectively. In addition to the total runtime, the runtime for the projection ($\pi$) and clustering ($\kappa$) steps are given. P-RASTER scales very well, while the parallel speedup of P-RASTER$'$ is more limited due to the cost incurred of having to join large hash tables at the end of the projection step.}
\label{table:par}
\end{table*}

\begin{proof}
Consider the base case of $n = 3$ adjacent clusters that are separated by two borders: Clusters $C_{\mathit{left}}$ and $C_{\mathit{middle}}$ are separated by border $B_1$; clusters $C_{\mathit{middle}}$ and $C_{\mathit{right}}$ are separated by border $B_2$. Assuming the algorithm proceeds from left to right, the first step results in a cluster $C_{\left|2\right|} = C_{\mathit{left}} \cup C_{\mathit{middle}}$ after removing $B_1$ as the clusters were adjacent and shared that very border. As $C_{\left|2\right|}$ still borders $B_2$, it remains a candidate cluster for joining. Subsequently, when processing the next border, the algorithm produces $C_{\left|3\right|} = C_{\left|2\right|} \cup C_{\mathit{right}}$, which is the final cluster. It should be easy to see that in the case of $k$ adjacent clusters that are separated by $k-1$ borders, the result is a cluster $C_{\left|k\right|}$. The inductive step is $C_{\left|k + 1\right|} = C_{\left|k\right|} \cup C_{\left|1\right|}$.
\end{proof}

The previous proof is also intuitive as clusters are processed from one vertical slice to the other and every newly encountered neighboring cluster entails the removal of one border and also means that the number of sub-clusters added increases by one. Theorems 1 and 2 work in tandem. For instance, if a cluster $C$ were to be split into more than three disjoint clusters around one border, Theorem 1 ensures that they are joined, while Theorem 2 ensures that clusters are joined across multiple borders. As a cluster remains a candidate for as long as it contains at least one tile that touches a border, it is irrelevant which case is addressed first.

P-RASTER slices the tile space vertically in a uniform manner. A practical rule of thumb is to set the number of slices to the number of available threads. The number of slices is only one aspect as it may also make sense to not slice uniformly but instead set the borders dynamically so that each slice contains approximately the same number of points, which cannot be expected when slicing the tile space uniformly with non-uniformly distributed data. What all these considerations have in common is that they depend on slicing not affecting the resulting clusters, which depends on the correctness of the next theorem.

\begin{lem}
Arbitrary slicing invariant: Slicing of the tile space can be arbitrary and does not affect the clusters in any way.
\end{lem}

\begin{proof}
The width of a slice is irrelevant. It only matters if a border divides a cluster or not. If it does not, then cluster reconstruction cannot be affected. On the other hand, if it does, then divided clusters will be unified based on Theorems 1 and 2, which do not depend on the location of the borders.
\end{proof}


\section{Runtime of P-RASTER}
\label{app-d}

As part of the evaluation presented in Sec.~\ref{raster-eval}, we benchmarked the performance of P-RASTER and P-RASTER$'$ with data sets containing 50M and 500M points. In Table~\ref{table:par}, detailed results are listed, showing the total runtime but also the runtime of the separate contraction/projection and clustering steps. We find that P-RASTER scales very well. The computationally intensive projection step is sped up by a factor of up to around 4 with 8 cores. This can be explained by the sequential bottleneck of joining hash tables that result from performing the projection in parallel on the various slices of the tile space. On the other hand, the computationally less expensive clustering step scales approximately with the number of cores. With 8 cores, it is up to 8 times faster. In contrast, the performance improvements with RASTER$'$ are more modest, but not insubstantial. While the clustering step scales similarly well, the projection step improves mildly with increasing CPU cores in the case of the 50M ($10^5$ clusters) input file. However, the picture with 500M points ($10^6$ clusters) is ambivalent, showing that with an increased tile space due to a precision value of 4, performance degrades when more than two cores are used.

\end{document}